\documentclass[10pt,compsoc]{IEEEtran}
\usepackage{algpseudocode}    
\usepackage{algorithm}
\usepackage{amsmath}
\usepackage{amsfonts}
\usepackage{booktabs}
\usepackage{multirow}
\usepackage{subfigure}
\usepackage{graphicx}         
\usepackage{color}
\usepackage{cite}             
\usepackage{comment}          
\usepackage{authblk}          
\usepackage{soul}             
\soulregister\cite7
\soulregister\ref7
\soulregister\pageref7
\usepackage{amsthm}

\newtheorem{myproblem}{\textbf{Problem}}
\newtheorem{mydefinition}{\textbf{Definition}}
\newtheorem{mytheorem}{\textbf{Theorem}}
\newtheorem{mylemma}{\textbf{Lemma}}
\newtheorem{myclaim}{\textbf{Claim}}

\makeatletter
\let\OldStatex\Statex
\renewcommand{\Statex}[1][3]{%
  \setlength\@tempdima{\algorithmicindent}%
  \OldStatex\hskip\dimexpr#1\@tempdima\relax
}
\makeatother

\graphicspath{{./figs/}}

\begin{document}


\title{
E-BLOW: E-Beam Lithography Overlapping aware Stencil Planning for MCC System
}

\author{
Bei Yu~\IEEEmembership{Member,~IEEE},
Kun Yuan,
Jhih-Rong Gao,
and David Z. Pan~\IEEEmembership{Fellow,~IEEE}
\thanks{The preliminary version has been presented at IEEE/ACM Design Automation Conference (DAC) in 2013.}
\thanks{B. Yu and D. Z. Pan are with the Department of Electrical and Computer Engineering, University of Texas, Austin, TX 78731 USA.}
\thanks{K. Yuan was with the Department of Electrical and Computer Engineering, University of Texas, Austin, TX 78731 USA.
He is now with Facebook Inc., Menlo Park, CA 94025 USA.}
\thanks{J-R. Gao was with the Department of Electrical and Computer Engineering, University of Texas, Austin, TX 78731 USA.
She is now with Cadence Design Systems, Austin, TX 78752 USA.}
}


\maketitle

\begin{abstract}
Electron beam lithography (EBL) is a promising maskless solution for the technology beyond 14nm logic node.
To overcome its throughput limitation, industry has proposed character projection (CP) technique, where some complex shapes (characters) can be printed in one shot.
Recently the traditional EBL system is extended into multi-column cell (MCC) system to further improve the throughput.
In MCC system, several independent CPs are used to further speed-up the writing process.
Because of the area constraint of stencil, MCC system needs to be packed/planned carefully to take advantage of the characters.
In this paper, we prove that the overlapping aware stencil planning ($\mathsf{OSP}$) problem is NP-hard.
To solve $\mathsf{OSP}$ problem in MCC system, we present a tool, E-BLOW, with several novel speedup techniques,
such as successive relaxation, dynamic programming, and KD-Tree based clustering. 
Experimental results show that, compared with previous works, E-BLOW demonstrates better performance for both conventional EBL system and MCC system.
\end{abstract}


\section{Introduction}
\label{sec:intro}
As the minimum feature size continues to scale to sub-22nm, the conventional 193nm optical photolithography technology is reaching its printability limit.
In the near future, multiple patterning lithography (MPL) has become one of the viable lithography techniques for 22nm and 14nm logic nodes
\cite{DPL_ICCAD08_Kahng,SADP_DAC2011_Zhang,TPL_ICCAD2011_Yu,MPL_DAC2014_Yu}.
In the longer future, i.e., for the logic nodes beyond 14nm,
extreme ultra violet (EUV), directed self-assembly (DSA), and electric beam lithography (EBL) are promising candidates as next generation lithography technologies \cite{LITH_TCAD2013_Pan}.
Currently, both EUV and DSA suffer from some technical barriers.
EUV technique is delayed due to tremendous technical issues such as lack of power sources, resists, and defect-free masks \cite{EUV_SPIE2010_Arisawa}.
DSA has only the potential to generate contact or via layers \cite{DSA_IEDM2010_Chang}.

EBL system, on the other hand, has been developed for several decades \cite{EBL_SPIE09_Pfeiffer}.
Compared with the traditional lithographic methodologies, EBL has several advantages.
(1) Electron beam can be easily focused into nanometer diameter with charged particle beam,
which can avoid suffering from the diffraction limitation of light.
(2) The price of a photomask set is getting unaffordable, especially through the emerging MPL techniques.
As a maskless technology, EBL can reduce the manufacturing cost. 
(3) EBL allows a great flexibility for fast turnaround times and even late design modifications to correct or adapt a given chip layout.
Because of all these advantages, EBL is being used in mask making, small volume LSI production, and R\&D to develop the technological nodes ahead of mass production.


Conventional EBL system applies variable shaped beam (VSB) technique.
In this mode, the entire layout is decomposed into a set of rectangles,
each being shot into resist by one electron beam.
In the printing process of VSB mode,
at first the electrical gun generates an initial beam, which becomes uniform through the shaping aperture.
Then the second aperture finalizes the target shape with a limited maximum size.
Since each pattern needs to be fractured into pieces of rectangles and printed one by one, the VSB mode suffers from serious throughput problem.

\begin{figure}[tb]
  \centering
  \includegraphics[width=0.4\textwidth]{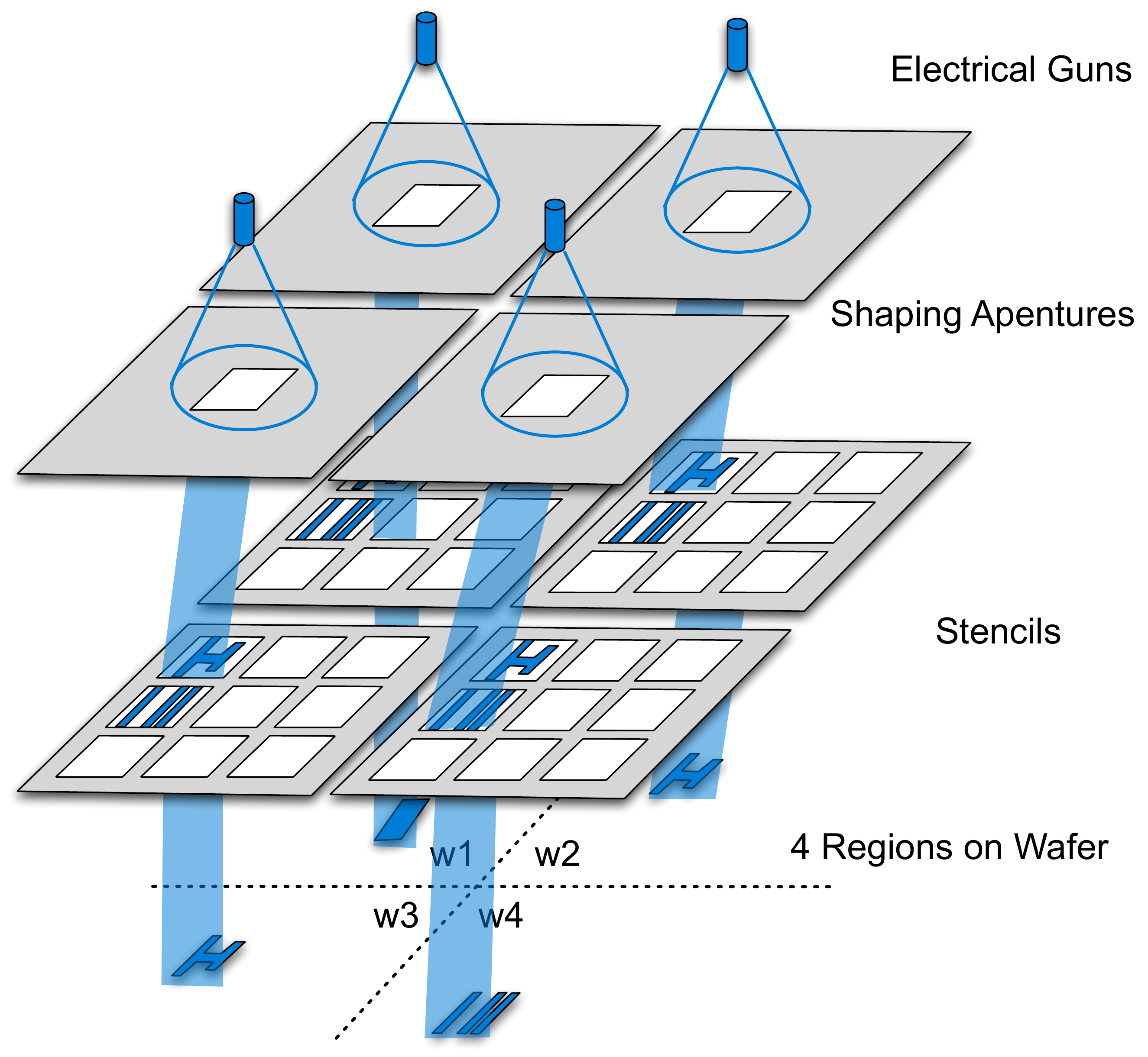}
  \caption{Printing process of MCC system, where four CPs are bundled.}
  \label{fig:projectMCC}
  \vspace{-.1in}
\end{figure}

One improved technique is called character projection (CP) \cite{EBL_SPIE2010_Fujimura},
where the second aperture is replaced by a \textit{stencil}.
Some complex shapes, called \textit{characters}, are prepared on the stencil.
The key idea is that if a pattern is pre-designed on the stencil, it can be printed in one electronic shot,
otherwise it needs to be fractured into a set of rectangles and printed one by one through VSB mode.
By this way the CP mode can improve the throughput significantly.
In addition, CP exposure has a good CD control stability compared with VSB \cite{EBL_SPIE08_Maruyama}.
However, the area constraint of stencil is the bottleneck. 
For modern design, due to the numerous distinct circuit patterns,  only limited number of patterns can be employed on stencil.
Those patterns not contained by stencil are still required to be written by VSB.
Thus one emerging challenge in CP mode is how to pack the characters into stencil to effectively improve the throughput.

Even with decades of development, the key limitation of the EBL system has been and still is the low throughput.
Recently, multi-column cell (MCC) system is proposed as an extension to CP technique \cite{EBL_SPIE04_Yasuda,EBL_SPIE2012_Maruyama}.
In MCC system, several independent character projections (CP) are used to further speed-up the writing process.
Each CP is applied on one section of wafer, and all CPs can work parallelly to achieve better throughput.
In morden MCC system, there are more than 1300 character projections (CPs) \cite{EBL_SPIE2010_Shoji}.
Since one CP is associated with one stencil, there are more than 1300 stencils in total.
The manufacturing of stencil is similar to mask manufacturing.
If each stencil is different, then the stencil preparation process would be very time consuming and expensive.
Due to the design complexity and cost consideration, different CPs share one stencil design.
One example of MCC printing process is illustrated in Fig. \ref{fig:projectMCC}, where four CPs are bundled to generate an MCC system.
In this example, the whole wafer is divided into four regions, $w_1, w_2, w_3$ and $w_4$, and each region is printed through one CP.
Note that the whole writing time of the MCC system is determined by the maximum one of the four regions.
For modern design, because of the numerous distinct circuit patterns, only limited number of patterns can be employed on stencil.
Since the area constraint of stencil is the bottleneck, the stencil should be carefully designed/manufactured to contain the most repeated cells or patterns.

Many previous works dealt with the design optimization for EBL system.
\cite{CEBL_ASPDAC2014_Gao,CEBL_ASPDAC2012_Du} considered EBL as a complementary lithography technique to print via/cut patterns.
\cite{EBL_SPIE03_Babin,EBL_TCAD2013_Fang} solved the subfield scheduling problem to reduce the critical dimension distortion.
\cite{EBL_SPIE06_Kahng,EBL_SPIE2011_Ma,EBL_ASPDAC2013_Yu} proposed a set of layout/mask fracturing approaches to reduce the VSB shot number.
Besides, several works solved the design challenges under CP technique.
\cite{EBL_ASPDAC2012_Du,EBL_ASPDAC2013_Ikeno} proposed several character design methods for both via layers and interconnect layers to achieve stencil area-efficiency.

\begin{figure} [tb]
  \centering
  \subfigure[]{\includegraphics[width=0.20\textwidth]{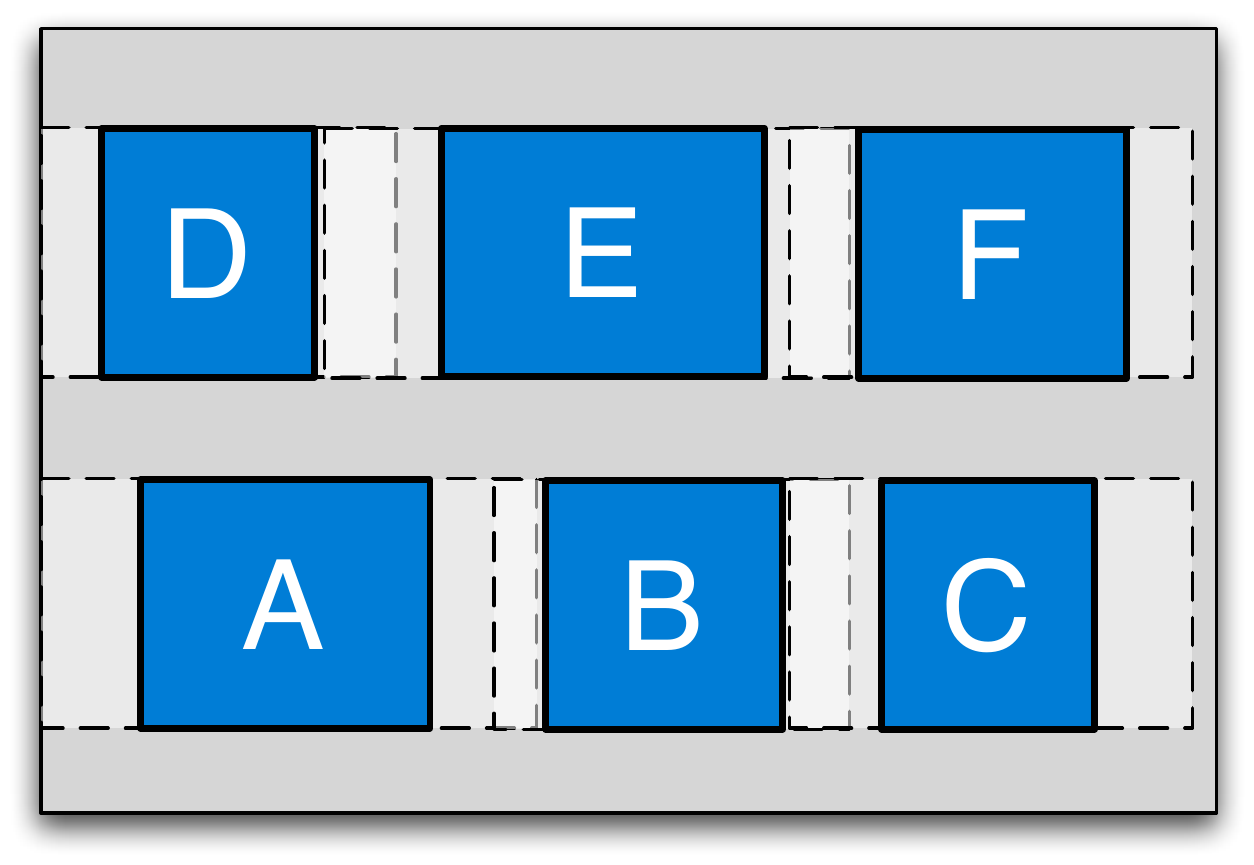}}
  \hspace{.2em}
  \subfigure[]{\includegraphics[width=0.20\textwidth]{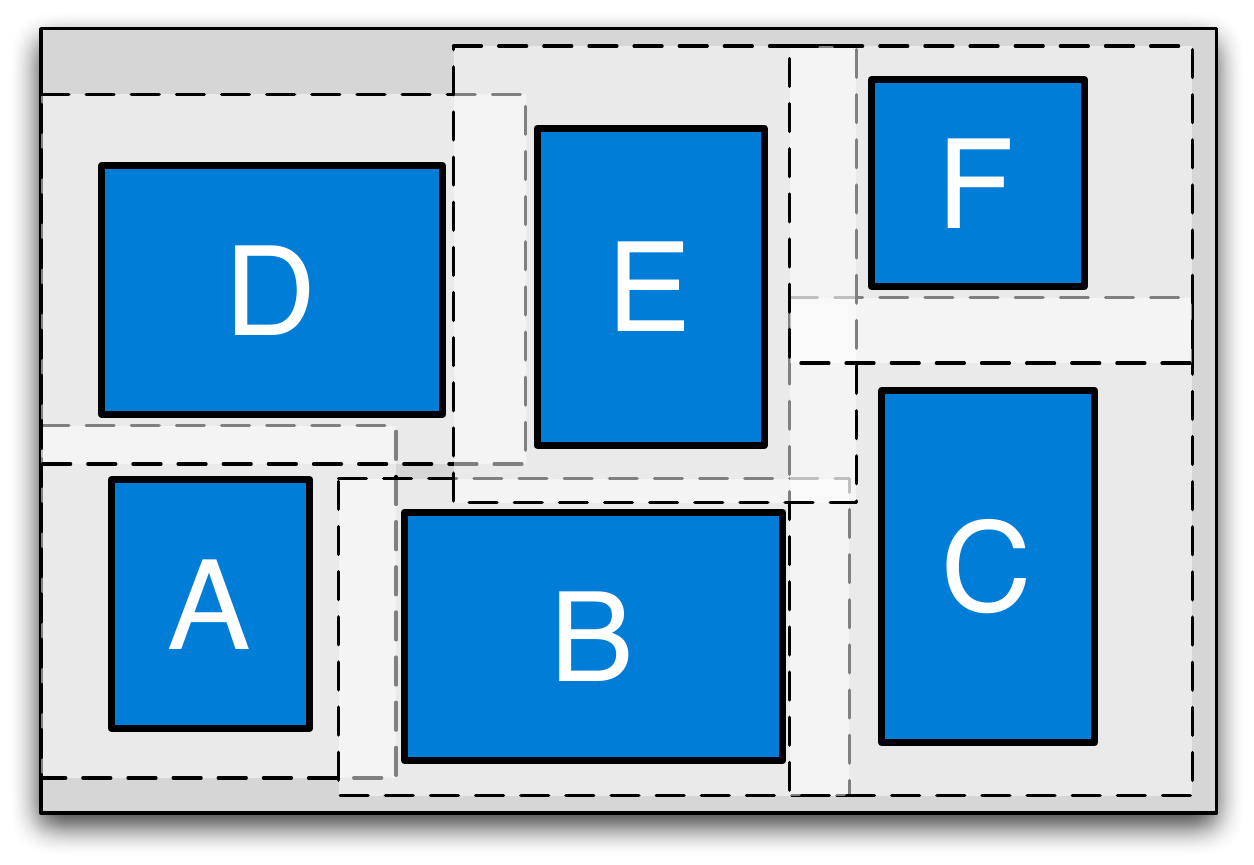}}
  \caption{Two types of $\mathsf{OSP}$ problem. (a) $\mathsf{1DOSP}$. (b) $\mathsf{2DOSP}$.}
  \label{fig:OSP}
  \vspace{-.1in}
\end{figure}

As one of the most challenges in CP mode, stencil planning has earned many attentions
\cite{EBL_IEICE06_Sugihara,EBL_TCAD2012_Yuan,EBL_ISPD2014_Kuang,EBL_ASPDAC2014_Chu,EBL_TCAD2014_Mak,EBL_ASPDAC2015_Guo}.
When blank overlapping is not considered, the stencil planning equals to a character selection problem.
\cite{EBL_IEICE06_Sugihara} proposed an integer linear programming (ILP) formulation to select a group of characters for throughput maximization.
When the characters can be overlapped to save more stencil space,
the corresponding stencil planning is referred as \textit{overlapping-aware stencil planning} ($\mathsf{OSP}$).
\cite{EBL_TCAD2012_Yuan,EBL_ISPD2014_Kuang} investigated on $\mathsf{OSP}$ problem to place more characters onto stencil.
Recently, \cite{EBL_ASPDAC2014_Chu,EBL_TCAD2014_Mak} assumed that the pattern position in each character can be shifted,
and integrated the character re-design into $\mathsf{OSP}$ problem.
As suggested in \cite{EBL_TCAD2012_Yuan}, the $\mathsf{OSP}$ problem can be divided into two types: $\mathsf{1DOSP}$ and $\mathsf{2DOSP}$.
In $\mathsf{1DOSP}$, the standard cells with same height are selected into stencil.
As shown in Fig.~\ref{fig:OSP}(a), each character implements one standard cell, and the enclosed circuit patterns of all the characters have the same height.
Note that here we only show the horizontal blanks, and the vertical blanks are not represented because they are identical.
In $\mathsf{2DOSP}$, the blank spaces of characters are non-uniform along both horizontal and vertical directions.
By this way, stencil can contain both complex via patterns and regular wires.
Fig.~\ref{fig:OSP}(b) illustrates a stencil design example for $\mathsf{2DOSP}$.

Compared with conventional EBL system, MCC system introduces two main challenges in $\mathsf{OSP}$ problem.
First, the objective is new: in MCC system the wafer is divided into several regions, and each region is written by one CP. Therefore the new $\mathsf{OSP}$ should minimize the maximal writing times of all regions.
However, in conventional EBL system the objective is simply minimize the wafer writing time.
Besides, the stencil for an MCC system can contain more than 4000 characters, previous methodologies for EBL system may suffer from runtime penalty.
However, no existing stencil planning work has been done toward the MCC system.

This paper presents E-BLOW, a comprehensive study to the MCC system $\mathsf{1DOSP}$ and $\mathsf{2DOSP}$ problems.
Our main contributions are summarized as follows.
\begin{itemize}
  \item We provide the proof that both $\mathsf{1DOSP}$ and $\mathsf{2DOSP}$ problems are NP-hard.
  \item We formulate integer linear programming (ILP) to co-optimizing characters selection and physical placements on stencil.
        To our best knowledge, this is the first mathematical formulation for both $\mathsf{1DOSP}$ and $\mathsf{2DOSP}$.
  \item We proposes a simplified formulation for $\mathsf{1DOSP}$. 
  \item We present a successive relaxation algorithm to find a near optimal solution.
  \item We design a KD-Tree based clustering algorithm to speedup $\mathsf{2DOSP}$ solution.
\end{itemize}

The rest of this paper is organized as follows.
Section \ref{sec:prelim} provides problem formulation.
Section \ref{sec:1d} presents algorithmic details to resolve $\mathsf{1DOSP}$ problem in E-BLOW,
while section \ref{sec:2d} details the E-BLOW solutions to $\mathsf{2DOSP}$ problem.
Section \ref{sec:result} reports experimental results, followed by the conclusion in Section \ref{sec:conclu}.

\section{Preliminaries}
\label{sec:prelim}

In this section, we provide the preliminaries regarding overlapping aware stencil planning ($\mathsf{OSP}$).
During character design, blank area is usually reserved around its boundaries.
Note in this paper, the blank space refers to the blank around the character boundaries.
The term ``overlapping'' means sharing blanks between adjacent characters.
By this way, more characters can be placed on the stencil \cite{EBL_TCAD2012_Yuan}.
In this section, first we will provide the detailed problem formulation, then we will prove that both $\mathsf{1DOSP}$ and $\mathsf{2DOSP}$ are NP-hard.

\subsection{Problem Formulation}
In an MCC system with $P$ CPs, the whole wafer is divided into $P$ regions $\{r_1, r_2, \dots, r_P\}$,
and each region is written by one particular CP.
We assume cell extraction \cite{EBL_SPIE09_Manakli} has been resolved first.
In other words, a set of character candidates $\{c_1, \cdots, c_n\}$ has already been given to the MCC system.
For each character candidate $c_i$, its writing time through VSB mode is denoted as $n_{i}$,
while its writing time through CP mode is $1$.

The regions of wafer have different layout patterns, and the throughputs would be also different.
Suppose character candidate $c_i$ repeats $t_{ic}$ times on region $r_c$.
Let $a_i$ indicate the selection of  character candidate $c_i$ as follows.
\begin{equation*}
	a_i = 
	\left\{
	\begin{array}{cc}
		1, 		& \textrm{candidate } c_i \textrm{ is selected on stencil}\\
		0,  	& \textrm{otherwise}
	\end{array}
	\right.
\end{equation*}
If $c_i$ is prepared on stencil, the total writing time of pattern $c_i$ on region $r_c$ is $t_{ic} \cdot 1$.
Otherwise, $c_i$ should be printed through VSB.
Since region $r_c$ comprises $t_{ic}$ candidate $c_i$, the writing time would be $t_{ic} \cdot n_i$.
Therefore, for region $r_c$ the total writing time $T_c$ is as follows:
\begin{eqnarray*}
  T_c & = & \sum_{i=1}^{n} a_i \cdot (t_{ic} \cdot 1) + \sum_{i=1}^{n} (1-a_i) \cdot (t_{ic} \cdot n_i) \\
      & = & \sum_{i=1}^n t_{ic} \cdot n_i - \sum_{i=1}^{n} t_{ic}  \cdot (n_i - 1) \cdot a_i                     \\
      & = & T_{c}^{VSB} - \sum_{i=1}^{n} R_{ic} \cdot a_i                                               
\end{eqnarray*}
where we denote $T_{c}^{VSB} = \sum_{i=1}^n t_{ic} \cdot n_i$, and $R_{ic} = t_{ic}  \cdot (n_i - 1)$.
$T_{c}^{VSB}$ shows the writing time on $r_c$ when only VSB is applied, 
and $R_{ic}$ represents the writing time reduction of candidate $c_i$ on region $r_c$.
In MCC system, for each region $r_c$ both $T_{c}^{VSB}$ and $R_{ic}$ are constants. 
Therefore, the total writing time of the MCC system is formulated as follows:
\begin{equation} \label{eq:obj}
\begin{split}
  T_{total} & = \textrm{max} \{T_c\}  \\
            & = \textrm{max} \{T_{c}^{VSB} - \sum_{i=1}^{n} R_{ic} \cdot a_i \},  \forall c \in P 
\end{split}
\end{equation}

Based on the notations above, we define the overlapping aware stencil planning ($\mathsf{OSP}$) for MCC system as follows.

\begin{myproblem}{\textbf{OSP for MCC System}}: 
Given a set of character candidate $C^{C}$, select a subset $C^{CP}$ out of $C^{C}$ as characters, and place them on the stencil.
The objective is to minimize the system writing time $T_{total}$ expressed by Eqn. (\ref{eq:obj}), while the placement of $C^{CP}$ is bounded by the outline of stencil.
The width and height of stencil is $W$ and $H$, respectively.
\end{myproblem}

For convenience, we use the term $\mathsf{OSP}$ to refer $\mathsf{OSP}$ for MCC system in the rest of this paper.

\subsection{NP-Hardness}
\label{sec:npc_proof}

In this subsection we will prove that both $\mathsf{1DOSP}$ and $\mathsf{2DOSP}$ are NP-hard.
To facilitate the proof, we first define a Bounded Subset Sum ($\mathsf{BSS}$) problem as follows.

\begin{myproblem}[{\bf Bounded Subset Sum}]
\label{prb:bss}
Given a list of $n$ numbers $x_1,\cdots, x_n$ and a number $s$,
where $\forall i \in [n] ~ 2\cdot x_i > x_{\max} ( \overset{\bigtriangleup}{=} \underset{i\in [n] }{\max}~|x_i|)$,
decide if there is a subset of the numbers that sums up to $s$.
\end{myproblem}

For example, given three numbers $1100, 1200, 1413$ and $T = 2300$, we can find a subset $\{1100, 1200\}$ such that $1100 + 1200 = 2300$.
Additionally, we can assumption that $t > c\cdot x_{\max}$, where $c$ is some constant.
Otherwise it be solved in $O(n^c)$ time.
Besides, without the bounded constraint $\forall i \in [ n] ~ 2\cdot x_i >  x_{\max}$, the $\mathsf{BSS}$ problem becomes {\bf Subset sum} problem,
which is in NP-complete \cite{NPC_B09_Arora}.
For simplicity of later explanation, let $S$ denote the set of $n$ numbers.
Note that, we can assume that all the numbers are integer numbers.

\begin{mytheorem}\label{thm:npc_bss}
    $\mathsf{BSS}$ problem is NP-complete.
\end{mytheorem}

The proof is in Appendix.
In the following, we will show that even a simpler version of $\mathsf{1DOSP}$ problem is NP-hard.
In this simpler version, there is only one row in the stencil, and a set of characters $C$ is given.
Besides, the blanks of each character are symmetric, and each character $c_i \in C$ is with the same length $w$.

\begin{mydefinition}[Minimum packing]
Given a subset of characters $C' \in C$, its minimum packing is the packing with the minimum stencil length.
\end{mydefinition}

\vspace{.1in}
\begin{mylemma}
\label{lem:symm}
Given a set of character $C = \{c_1, c_2, \dots, c_n\}$ placed on a single row stencil.
If for each character $c_i \in C$, both of its left and right blanks are $s_i$, then the minimum packing is with the following stencil length
\begin{equation}\label{eq:symm_len}
    n \cdot M - \overset{n}{\underset{i=1}{\sum}} s_i + \max_{i \in [n]} \{s_i\}
\end{equation}
\end{mylemma}

\begin{proof}
Without loss of generality, we assume that $s_1 \ge s_2 \ge \dots \ge s_n$.
We prove by recursion that in an minimum length packing, the overlapping blank is $f(n) = \sum_{i=2}^n s_i$.
If there are only two characters, it is trivial that $f(2) = s_2$.
We assume that when $p=n-1$, the maximum overlapping blank $f(n-1) = \sum_{i=2}^{n-1} s_i$.
For the last character $c_n$, the maximum sharing blank value is $s_n$.
Since for any $i<n$, $s_i \ge s_n$, we can simply insert it at either the left end or the right end, and find the incremental overlapping blank $s_n$.
Thus $f(n) = f(n-1) + s_n = \sum_{i=2}^n s_i$.
Because the maximum overlapping blank for all characters is $\sum_{i=2}^n s_i$, we can see the minimum packing length is as in Eqn. (\ref{eq:symm_len}).
\end{proof}
\vspace{.1in}

\begin{mylemma}\label{lem:1dosp}
$\mathsf{BSS} \leq_p \mathsf{1DOSP}$.
\end{mylemma}

\begin{proof}
Given an instance of $\mathsf{BSS}$ with $s$ and $S = \{x_1, x_2, \dots, x_n\}$, we construct a $\mathsf{1DOSP}$ instance as follows:
\begin{itemize}
  \item
  The stencil length is set to $M+s$, where $M=\max_{i \in [n]}\{x_i\}$.
  \item
  For each $x_i \in S'$, in $\mathsf{1DOSP}$ there is a character $c_i$, whose width is $M$ and both left and right blanks are $M-x_i$.
  Since $x_i > M/2$, the sum of left blank and right blank is less or equal to $M$.
  \item
  We introduce an additional character $c_0$, whose width size is $M$, and both left and right blanks are $M-\min_{i \in [n]}\{x_i\}$.
  \item
  The VSB writing time of character $c_0$ is set to $\sum_{i \in [n]} x_i$, while the VSB writing time for each character $c_i$ is set to $x_i$.
  The CP writing times are set to 0 for all characters.
  \item
  There is only one region, and each character $c_i$ repeats one time in the region.
\end{itemize}

For instance, given initial set $S = \{1100, 1200, 2000\}$ and $s = 2300$,
the constructed $\mathsf{1DOSP}$ instance is shown in Fig. \ref{fig:1dosp_instance}.

\begin{figure}[tb!]
    \centering
    \subfigure[]{\includegraphics[width=0.304\textwidth]{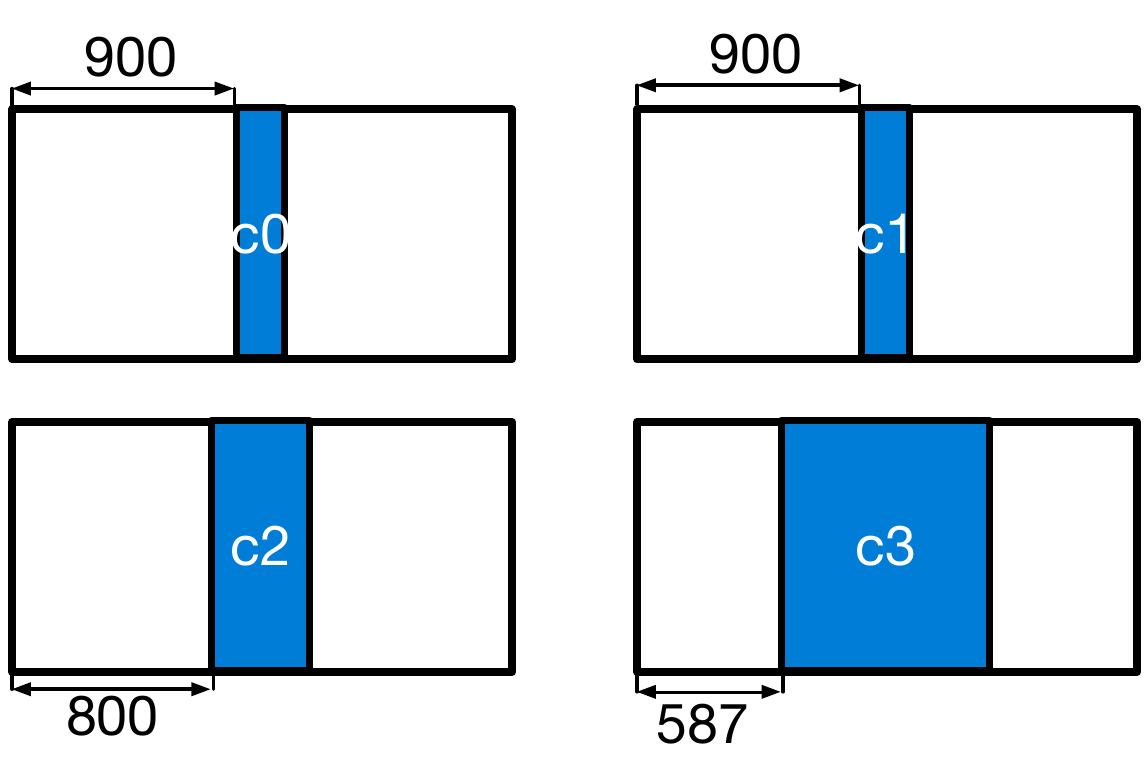}}
    \hspace{.2in}
    \subfigure[]{\includegraphics[width=0.34\textwidth]{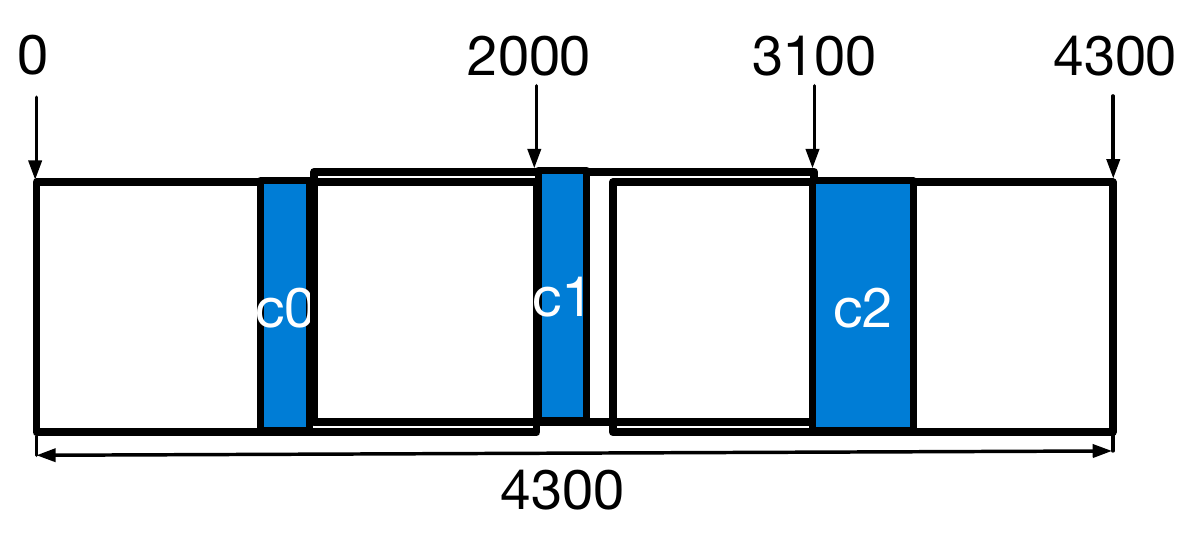}}
    \caption{
    (a) $\mathsf{1DOSP}$ instance for the $\mathsf{BSS}$ instance $S = \{1100, 1200, 2000\}$ and $s = 2300$.
    (b) The minimum packing is with stencil length $M + s = 2000 + 2300 = 4300$.
    }
    \label{fig:1dosp_instance}
    \vspace{-.1in}
\end{figure}

We will show the $\mathsf{BSS}$ instance $S = \{x_1, x_2, \dots, x_n\}$ has a subset that adds up to $s$ if and only if the constructed $\mathsf{1DOSP}$ instance has minimum packing length $M+s$ and total writing time smaller than $\sum x_i$.

($\Rightarrow$ part)
After solving the $\mathsf{BSS}$ problem, a set of items $S'$ are selected that they add up to $s$.
For each $x_i \in S'$, character $c_i$ is also selected into the stencil.
Besides, since the system writing time for $c_0$ is $\sum x_i$, it is trivial to see that in the $\mathsf{1DOSP}$ instance the $c_0$ must be selected.
Due to the Lemma \ref{lem:symm}, the minimum total packing length is
\begin{equation*}
  (n+1) \cdot M - \sum_{i \in S'} (M-x_i) = M + \sum_{i \in S'} x_i = M + s
\end{equation*}
Meanwhile, the minimum total writing time in the $\mathsf{1DOSP}$ is $\sum x_i - s$.

($\Leftarrow$ part)
We start from a $\mathsf{1DOSP}$ instance with minimum packing length $M+s$ and total writing time smaller than $\sum x_i$, where a set of character $C' \in C$ are selected.
Since the total total writing time must be smaller than $\sum x_i$, character $c_0 \in C'$.
For all characters in set $c_i \in C'$ except $c_0$, we select $x_i$ into the subset $S' \in S$, which adds up to $s$.

\end{proof}

\begin{mytheorem}
\label{thm:nph_1d}
$\mathsf{1DOSP}$ is in NP-hard.
\end{mytheorem}

\begin{proof}
Directly from Lemma \ref{lem:1dosp} and Theorem \ref{thm:npc_bss}.
\end{proof}

\begin{mytheorem}
\label{thm:nph_2d}
$\mathsf{2DOSP}$ is in NP-hard.
\end{mytheorem}

Since $\mathsf{1DOSP}$ is a special case of $\mathsf{2DOSP}$.
Due to the NP-hardness of $\mathsf{1DOSP}$, the $\mathsf{2DOSP}$ problem is NP-hard as well.
Combining Theorem \ref{thm:nph_1d} and Theorem \ref{thm:nph_2d},
we can achieve the conclusion that $\mathsf{OSP}$ problem, even for conventional EBL system, is NP-hard.

\section{E-BLOW for $\mathsf{1DOSP}$}
\label{sec:1d}

When each character implements one standard cell, the enclosed circuit patterns of all the characters have the same height.
The corresponding $\mathsf{OSP}$ problem is called $\mathsf{1DOSP}$, which can be viewed as a combination of character selection and single row ordering problems
\cite{EBL_TCAD2012_Yuan}.
Different from two-step heuristic proposed in \cite{EBL_TCAD2012_Yuan},
we show that these two problems can be solved simultaneously through a unified ILP formulation (\ref{eq:1ilp}).
For convenience, Table \ref{tab:notations1} lists the notations used in $\mathsf{1DOSP}$ problem.

\begin{table}[tb!]
\centering 
\caption{Notations used in 1D-ILP Formulation}
\label{tab:notations1}
\resizebox{7.4cm}{!}{
  \begin{tabular}{|l|l|}
  \hline
  \hline   $W$             & width constraint of stencil or row \\
  \hline   $n$             & number of characters \\
  \hline   $m$             & number of rows \\
  \hline   $x_i$           & x-position of character $c_i$ \\
  \hline   $w_i$           & width of character $c_i$ \\
  \hline   $o_{ij}^h$      & horizontal overlap between $c_i$ and $c_j$ \\
  \hline   $p_{ij}$        & 0-1 variable, $p_{ij} = 0$ if $c_i$ is left of $c_j$\\
  \hline   $a_{ij}$        & 0-1 variable, $a_{ij} = 1$ if $c_i$ is on $j$th row\\
  \hline \hline
  \end{tabular}
}
\end{table}

\begin{figure}[htb]
{
\small
\begin{align}
 \textrm{min} \ & \ \ T_{total}  & \label{eq:1ilp}\\
 \textrm{s.t} \ \  
 & T_{total} \ge T_{c}^{VSB} - \sum_{i=1}^{n}(\sum_{k=1}^{m} R_{ic} \cdot a_{ik}) & \forall c \in P   \label{1ilp_a}\tag{\ref{eq:1ilp}$a$}\\
 & 0 \le x_i \le W - w_i                                                          & \forall i         \label{1ilp_b}\tag{\ref{eq:1ilp}$b$}\\
 & \sum_{k=1}^m a_{ik} \le 1                                                      & \forall i         \label{1ilp_c}\tag{\ref{eq:1ilp}$c$}\\
 & x_i + w_{ij} - x_j \le W (2 + p_{ij} - a_{ik} - a_{jk})                        & \forall i,j	      \label{1ilp_d}\tag{\ref{eq:1ilp}$d$}\\
 & x_j + w_{ji} - x_i \le W (3 - p_{ij} - a_{ik} - a_{jk})                        & \forall i,j	      \label{1ilp_e}\tag{\ref{eq:1ilp}$e$}\\
 & a_{ik}, a_{jk}, p_{ij}: 0-1 \ \textrm{variable}                                & \forall i,j	      \label{1ilp_f}\tag{\ref{eq:1ilp}$f$}
\end{align}
}
\vspace{-.2in}
\end{figure}

In formulation (\ref{eq:1ilp}), $W$ is the stencil width, $m$ is the number of rows.
For each character $c_i$, $w_i$ and $x_i$ are the width and the x-position, respectively.
If and only if $c_i$ is assigned to $k$-th row, $a_{ik}=1$. Otherwise, $a_{ik}=0$.
Constraints (\ref{1ilp_d}) (\ref{1ilp_e}) are used to check position relationship between $c_i$ and $c_j$.
Here $w_{ij} = w_i -o_{ij}^h$ and
$w_{ji} = w_j -o_{ji}^h$,
where $o_{ij}^h$ is the overlapping when candidates $c_i$ and $c_j$ are packed together.
Only when $a_{ik} = a_{jk} = 1$, i.e. both character $i$ and character $j$ are assigned to row $k$,
one of the two constraints (\ref{1ilp_d}) (\ref{1ilp_e}) will be active.
Besides, for any three characters $c_1, c_2, c_3$ being assigned to row $k$, i.e., $a_{1k}=a_{2k}=a_{3k}=1$,
the $p_{12}, p_{13}$ and $p_{23}$ are self-consistent.
That is, if $c_1$ is on the left of $c_2$ ($p_{12}=0$) and $c_2$ is on the left of $c_3$ ($p_{23}=0$),
then $c_1$ should be on the left of $c_3$ ($p_{13}=0$).
Similarly, if $c_1$ is on the right of $c_2$ ($p_{12}=1$) and $c_2$ is on the right of $c_3$ ($p_{23}=1$),
then $c_1$ should be on the right of $c_3$ ($p_{13}=1$) as well.

Since ILP is a well known NP-hard problem, directly solving it may suffer from long runtime penalty. 
One straightforward speedup method is to relax the ILP into the corresponding linear programming (LP) through replacing constraints (\ref{1ilp_f}) by the following:
\begin{equation*}
  0 \le a_{ik}, a_{jk}, p_{ij} \le 1
\end{equation*}

It is obvious that the LP solution provides a lower bound to the ILP solution.
However, we observe that the solution of relaxed LP could be like this:
for each $i$, $\sum_j a_{ij}=1$ and all the $p_{ij}$ are assigned $0.5$.
Although the objective function is minimized and all the constraints are satisfied,
this LP relaxation provides no useful information to guide future rounding, i.e.,
all the character candidates are selected and no ordering relationship is determined.

\begin{figure} [tb]
  \centering
  \includegraphics[width=0.44\textwidth]{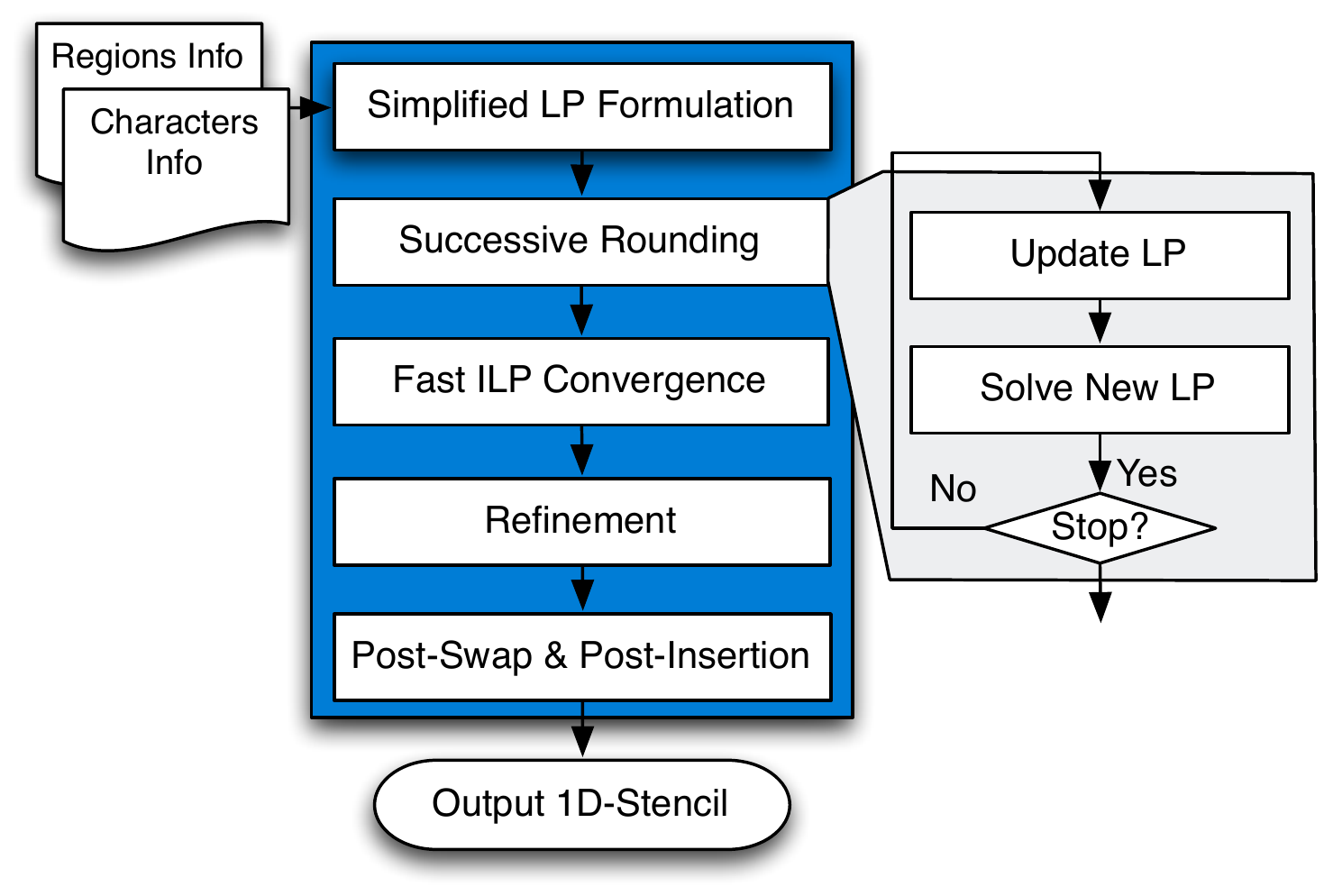}
  \caption{E-BLOW overall flow for $\mathsf{1DOSP}$.}
  \label{fig:Flow1D}
  \vspace{-.1in}
\end{figure}

To overcome the limitation of above rounding, E-BLOW proposes a novel successive rounding framework to search near-optimal solution in reasonable runtime.
As shown in Fig. \ref{fig:Flow1D}, the overall flow includes several steps:
Simplified ILP formulation, Successive Rounding, Fast ILP Convergence, Refinement, Post-swap and Post-Insertion.
In section \ref{sec:lp} the simplified formulation will be discussed, and its LP rounding lower bound will be proved.
In section \ref{sec:succ} the details of successive rounding would be introduced.
In section \ref{sec:ilp} the Fast ILP convergence technique would be presented.
In section \ref{sec:refine} the refinement process is proposed.
At last, to further improve the performance, in section \ref{sec:post} the post-swap and post-insertion techniques are discussed.

\subsection{Simplified ILP Formulation}
\label{sec:lp}

As discussed above, solving the ILP formulation (\ref{eq:1ilp}) is very time consuming, and the related LP relaxation may be bad in performance.
To overcome the limitations of (\ref{eq:1ilp}), in this section we introduce a simplified ILP formulation, whose LP relaxation can provide good lower bound.
The simplified formulation is based on a \textit{symmetrical blank} (S-Blank) assumption:
the blanks of each character are symmetric, i.e., left blank equals to right blank.
$s_i$ is used to denote the blank of character $c_i$.
Note that for different characters $c_i$ and $c_j$, their blanks $s_i$ and $s_j$ can be different.

At first glance the S-Blank assumption may lose optimality.
However, it provides several practical and theoretical benefits.
(1) In \cite{EBL_TCAD2012_Yuan} the single row ordering problem was transferred into Hamilton Cycle problem, 
which is a well known NP-hard problem and even particular solver is quite expensive.
In our work, instead of relying on expensive solver, under this assumption the problem can be optimally solved in $O(n)$.
(2) Under S-Blank assumption, the ILP formulation can be effectively simplified to provide a reasonable rounding bound theoretically.
Compared with previous heuristic framework \cite{EBL_TCAD2012_Yuan}, the proved rounding bound provides a better guideline for a global view search.
(3) To compensate the inaccuracy in the asymmetrical blank cases, E-BLOW provides a refinement (see section \ref{sec:refine}).

%

The simplified ILP formulation is shown in Eqn. (\ref{fast}).

\begin{figure}
\begin{align}
    \textrm{max}  & \ \sum_i \sum_j  a_{ij} \cdot profit_i   &               \label{fast}\\
    \textrm{s.t.} \ \
    & \sum_i (w_i - s_i) \cdot a_{ij} \le W - B_j            & \forall j     \label{fast_a}\tag{\ref{fast}$a$}\\
    & B_j \ge s_i \cdot a_{ij}                               & \forall i,j   \label{fast_b}\tag{\ref{fast}$b$}\\
    & \sum_j a_{ij}  \le 1                                   & \forall i     \label{fast_c}\tag{\ref{fast}$c$}\\
    & a_{ij} = 0 \ \textrm{or}\ 1                            & \forall i,j   \label{fast_d}\tag{\ref{fast}$d$}
\end{align}
\vspace{-.2in}
\end{figure}

In the objective function of Eqn. (\ref{fast}), each character $c_i$ is associated with one profit value $profit_i$.
The $profit_i$ value is to evaluate the overall system writing time improvement if character $c_i$ is selected.
Through assigning each character $c_i$ with one profit value, we can simplify the complex constraint (\ref{1ilp_a}).
More details regarding the profit value setting would be discussed in Section \ref{sec:succ}.
Besides, due to Lemma \ref{lem:symm}, constraint (\ref{fast_a}) and constraint (\ref{fast_b}) are for row width calculation,
where (\ref{fast_b}) is to linearize $max$ operation.
Here $B_j$ can be viewed as the maximum blank space of all the characters on row $r_j$.
Constraint (\ref{fast_c}) implies each character can be assigned into at most one row.
It's easy to see that the number of variables is $O(nm)$, where $n$ is the number of characters, and $m$ is the number of rows.
Generally speaking, single character number $n$ is much larger than row number $m$.
Thus compared with basic ILP formulation (\ref{eq:1ilp}), the variable number in (\ref{fast}) can be reduced dramatically.

In our implementation, we set $s_i$ to $\lceil (sl_i + sr_i ) / 2 \rceil$, where $sl_i$ and $sr_i$ are $c_i$'s left blank and right blank, respectively.
Note that here the ceiling function is used to make sure that under the S-Blank assumption, each blank is still integral.
Although this setting may loss some optimality,  E-BLOW provides post-stage to compensate the inaccuracy through incremental character insertion.

Now we will show that the LP relaxation of (\ref{fast}) has reasonable lower bound.
To explain this, let us first look at a similar formulation (\ref{knapsack}) as follows:
\begin{align}
    \textrm{max}  & \ \sum_i \sum_j a_{ij} \cdot profit_i     &              \label{knapsack}\\
    \textrm{s.t.} \ \ 
    & \sum_i (w_i - s_i) \cdot a_{ij} \le W - max_s      	  & \forall j    \label{knapsack_a}\tag{\ref{knapsack}$a$}\\
    &  (4c) - (4d)                                            &              \notag
\end{align}
where $max_s$ is the maximum horizontal blank length of every character, i.e. $max_s = \textrm{max}\{s_i | i = 1, 2, \dots, n\}$.
Program (\ref{knapsack}) is a multiple knapsack problem \cite{Knapsack_b90_Martello}.
A multiple knapsack is similar to a knapsack problem, with the difference that there are multiple knapsacks.
In formulation (\ref{knapsack}), each $profit_i$ can be rephrased as $(w_i-s_i) \times ratio_i$.

\begin{mylemma}
\label{lem:half_bound}
If each $ratio_i$ is the same, the multiple knapsack problem (\ref{knapsack}) can find a $0.5-$approximation algorithm using LP rounding method.
\end{mylemma}

For brevity we omit the proof, detailed explanations can be found in \cite{Knapsack_JCO00_Dawande}.
When all $ratio_i$ are the same, formulation (\ref{knapsack}) can be approximated to a max-flow problem.
In addition, if we denote $\alpha$ as $\textrm{min} \{ratio_i\}$ /$\textrm{max} \{ratio_i\}$, we can achieve the following Lemma:

\begin{mylemma}
\label{lem:alpha_bound}
The LP rounding solution of (\ref{knapsack}) can be a $0.5\alpha-$ approximation to optimal solution of (\ref{knapsack}).
\end{mylemma}

\begin{proof}
First we introduce a modified formulation to program (\ref{knapsack}), where each $profit_i$ is set to min$\{profit_i\}$.
In other words, in the modified formulation, each $ratio_i$ is the same.
Let $OPT$ and $OPT'$ be the optimal values of (\ref{knapsack}) and the modified formulation, respectively.
Let $APR'$ be the corresponding LP rounding result in the modified formulation.
According to Lemma \ref{lem:half_bound}, $APR' \ge 0.5 \cdot OPT'$.
Since $\textrm{min}\{profit_i\} \ge profit_i \cdot \alpha$, we can get $OPT' \ge \alpha \cdot OPT$.
In summary, $APR' \ge 0.5 \cdot OPT' \ge 0.5\alpha \cdot OPT$.
\end{proof}

The difference between (\ref{fast}) and (\ref{knapsack}) is the right side values at (\ref{fast_a}) and (\ref{knapsack_a}).
Blank spacing is relatively small comparing with the row length, we can get that $W - max_s \approx W - B_j$.
Then we can expect that program (\ref{fast}) has a reasonable rounding performance.


\subsection{Successive Rounding}
\label{sec:succ}

In this subsection we propose a successive rounding algorithm to solve program (\ref{fast}) iteratively.
Successive rounding uses a simple iterative scheme in which fractional variables are rounded one after the other until an integral solution is found \cite{ILP_IJOC00_Johnson}.
The ILP formulation (\ref{fast}) becomes an LP if we relax the discrete constraint to a continuous constraint as:
$0 \le a_{ij} \le 1$.

\begin{algorithm}[tb!]
\caption{SuccRounding ( $th_{inv}$ )}
\label{alg:round}
\begin{algorithmic}[1]
  \Require{ILP Formulation (\ref{fast})}
  \State  Set all $a_{ij}$ as unsolved;
  \Repeat
    \State Update $profit_i$ for all unsolved $a_{ij}$;
    \State Solve relaxed LP of (\ref{fast});
    \Repeat
      \State  $a_{pq} \leftarrow$ max$\{a_{ij}$\};
      \ForAll {$a_{ij} \ge a_{pq} \times th_{inv}$ }
        \If {$c_i$ can be assigned to row $r_j$}
          \State $a_{ij} = 1$ and set it as solved;
          \State Update capacity of row $r_j$;
        \EndIf
      \EndFor
     \Until  {cannot find $a_{pq}$}
  \Until  { }
\end{algorithmic}
\end{algorithm}

The details of successive rounding is shown in Algorithm \ref{alg:round}.
At first we set all $a_{ij}$ as \textit{unsolved} since none of them is assigned to rows.
The LP is updated and solved iteratively.
For each new LP solution, we search the maximal $a_{ij}$, and store in $a_{pq}$ (line 6).
Then we find all $a_{ij}$ that is closest to the maximum value $a_{pq}$, i.e., $a_{ij} \ge a_{pq} \times th_{inv}$.
In our implementation, $th_{inv}$ is set to 0.9.
For each selected variables $a_{ij}$, we try to pack $c_i$ into row $r_j$, and set $a_{ij}$ as \textit{solved}.
Note that when one character $c_i$ is assigned to one row, all $a_{ij}$ would be set as solved.
Therefore, the variable number in updated LP formulation would continue to decrease.
This procedure repeats until no appropriate $a_{ij}$ can be found.
One key step of Algorithm \ref{alg:round} is the $profit_i$ update (line 3).
For each character $c_i$, we set its $profit_i$ as follows:
\begin{equation}
  profit_i = \sum_c \frac{t_{c}}{t_{max}} \cdot (n_i - 1) \cdot t_{ic} \label{eq:profit}
\end{equation}
where $t_c$ is current writing time of region $r_c$, and $t_{max} =$ max $\{t_c, \forall c \in P\}$.
Through applying the $profit_i$, the region $r_c$ with longer writing time would be considered more during the LP formulation. 
During successive rounding, if $c_i$ is not assigned to any row, $profit_i$ would continue to be updated, so that the total writing time of the whole MCC system can be minimized.

\subsection{Fast ILP Convergence}
\label{sec:ilp}
\begin{figure}[htb]
  \centering
  \includegraphics[width=0.44\textwidth]{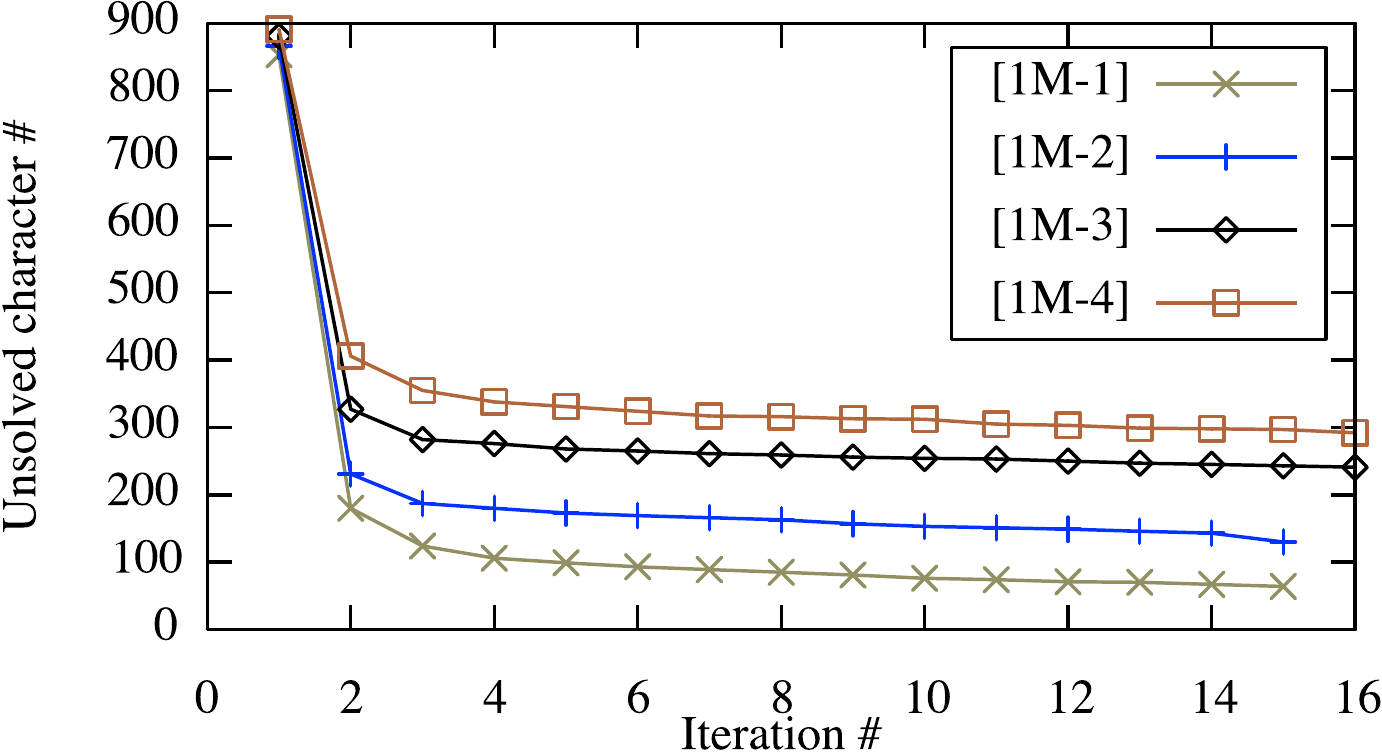}
  \caption{Unsolved character number along the LP iterations for testcases 1M-1, 1M-2, 1M-3, and 1M-4.}
  \label{fig:lp_iteration}
  \vspace{-.1in}
\end{figure}

During successive rounding, for each LP iteration, we select some characters into rows, and set these characters as solved.
In the next LP iteration, only unsolved characters would be considered in formulation. 
Thus the number of unsolved characters continues to decrease through the iterations.
For four test cases (1M-1 to 1M-4), Fig. \ref{fig:lp_iteration} illustrates the number of unsolved characters in each iteration.
We observe that in early iterations, more characters would be assigned to rows.
However, when the stencil is almost full, fewer of $a_{ij}$ could be close to $1$.
Thus,
in late iterations only few characters would be assigned into stencil, and the successive rounding requires more iterations.

\begin{algorithm}[tb!]
\caption{Fast ILP Convergence ( $L_{th}, U_{th}$ )}
\label{alg:ilp}
\begin{algorithmic}[1]
  \Require{Solutions of relaxed LP (\ref{fast});}
  \ForAll {$a_{ij}$ in relaxed LP solutions}
    \If {$a_{ij} < L_{th}$}
      \State Set $a_{ij}$ as solved;
    \EndIf
    \If {$a_{ij} > U_{th}$}
      \State Assign $c_i$ to row $r_j$;
      \State Set $a_{ij}$ as solved;
    \EndIf
  \EndFor
  \State Solve ILP formulation (\ref{fast}) for all unsolved $a_{ij}$
  \If {$a_{ij} = 1$}
    \State Assign $c_i$ to row $r_j$;
  \EndIf
\end{algorithmic}
\end{algorithm}

To overcome this limitation so that the successive rounding iteration number can be reduced,
we present a convergence technique based on fast ILP formulation.
The basic idea is that when we observe only few characters are assigned into rows in one LP iteration, we stop successive rounding in advance,
and call fast ILP convergence to assign all left characters.
Note that in \cite{EBL_ISPD2014_Kuang} an ILP formulation with similar idea was also applied.
The details of the ILP convergence is shown in Algorithm \ref{alg:ilp}.
The input are the solutions of last LP rounding, and two parameters $L_{th}$ and $U_{th}$.
First we check each $a_{ij}$ (lines 1-9).
If $a_{ij} < L_{th}$, then we assume character $c_i$ would be not assigned to row $r_j$, and set $a_{ij}$ as solved.
Similarly, if $a_{ij} > U_{th}$, we assign $c_i$ to row $r_j$ and set $a_{ij}$ as solved.
For those unsolved $a_{ij}$ we build up ILP formulation (\ref{fast}) to assign final rows (lines 10-13).

\begin{figure}[tb]
  \centering
  \includegraphics[width=0.44\textwidth]{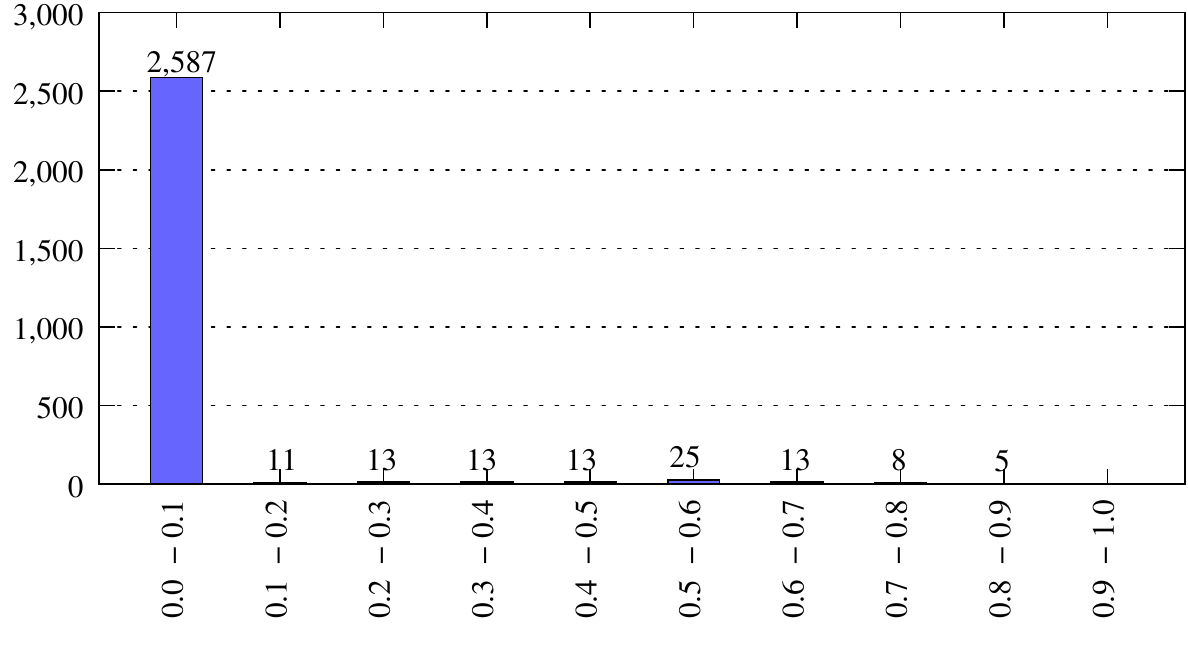}
  \vspace{-.1in}
  \caption{For test case 1M-1, solution distribution in last LP, where most of values are close to 0.}
  \label{fig:distribution_1M1}
  \vspace{-.1in}
\end{figure}

At first glance the ILP formulation may be expensive to solve.
However, we observe that in our convergence Algorithm \ref{alg:ilp}, typically the variable number is small.
Fig. \ref{fig:distribution_1M1} illustrates the solution distribution in last LP formulation.
We can see that most of the values are close to 0.
In our implementation $L_{th}$ and $U_{th}$ are set to 0.1 and 0.9, respectively.
For this case, although the LP formulation contains more than 2500 variables, our fast ILP formulation results in only 101 binary variables.

\subsection{Refinement}
\label{sec:refine}

Refinement is a stage to solve the single row ordering problem \cite{EBL_TCAD2012_Yuan}, which adjusts the relative locations of input $p$ characters to minimize the total width.
Under the S-Blank assumption, because of Lemma \ref{lem:symm}, this problem can be optimally solved through the following two-step greedy approach.
\begin{enumerate}
  \item All characters are sorted decreasingly by blanks;
  \item All characters are inserted one by one. Each one can be inserted at either left end or right end.
\end{enumerate}

\begin{figure} [tb]
  \centering
  \subfigure[]{\includegraphics[width=0.28\textwidth]{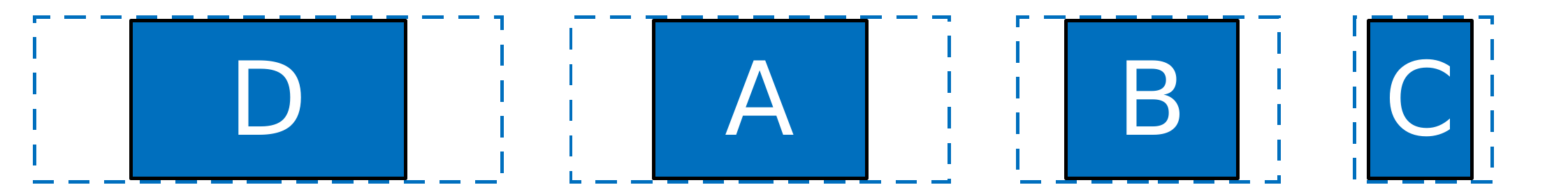}}
  \hspace{.3in}
  \subfigure[]{\includegraphics[width=0.15\textwidth]{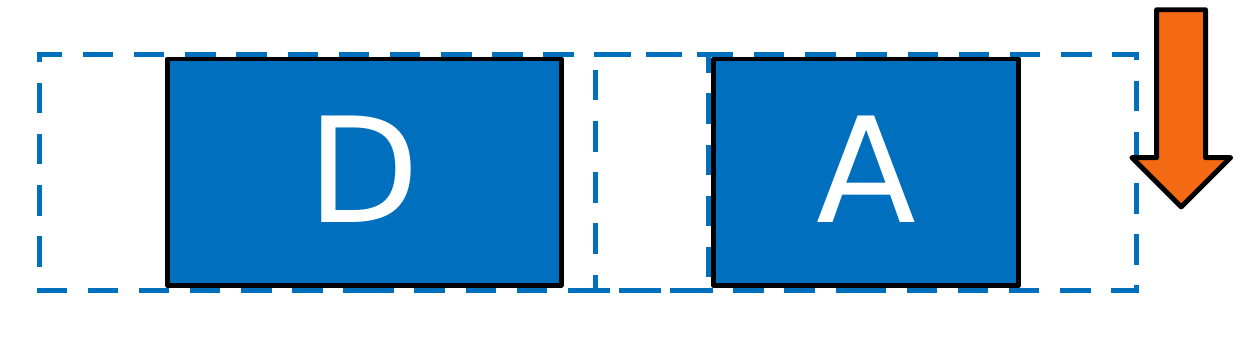}}
  \subfigure[]{\includegraphics[width=0.22\textwidth]{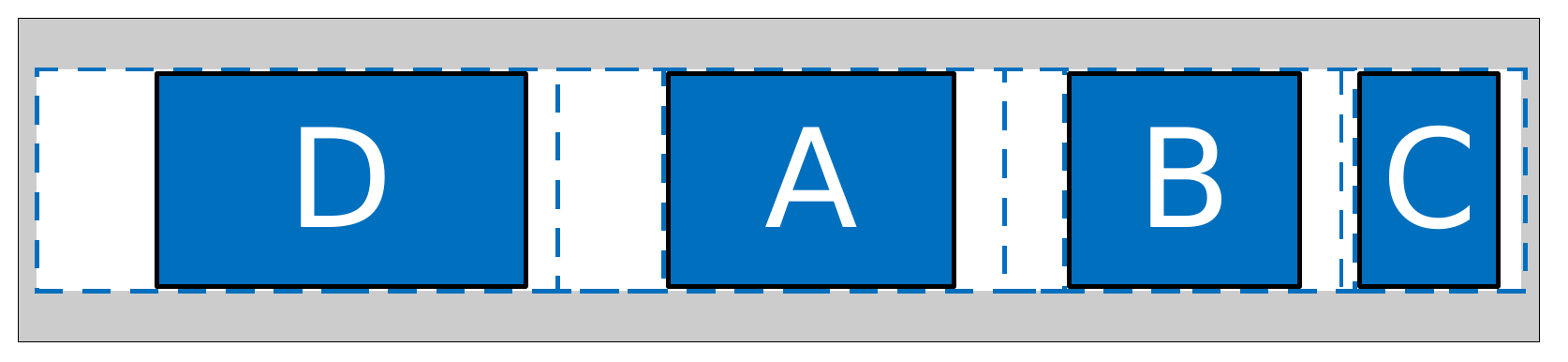}}
  \hspace{.1in}
  \subfigure[]{\includegraphics[width=0.15\textwidth]{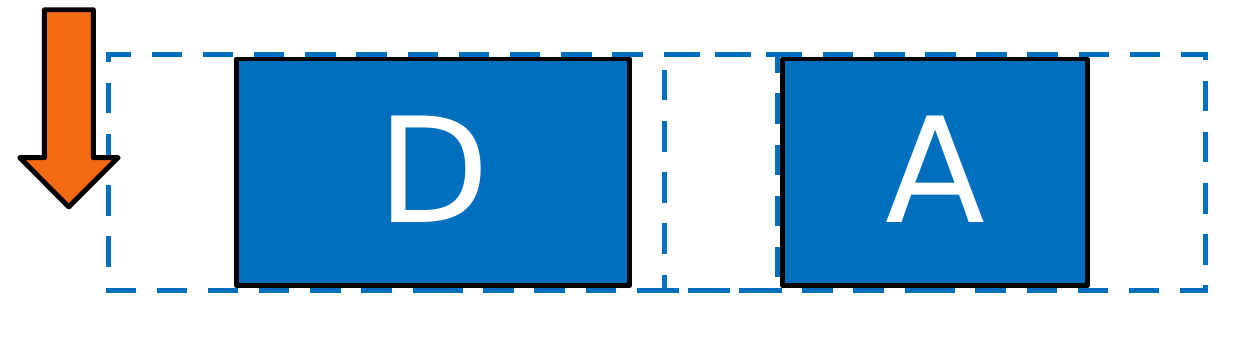}}
  \subfigure[]{\includegraphics[width=0.22\textwidth]{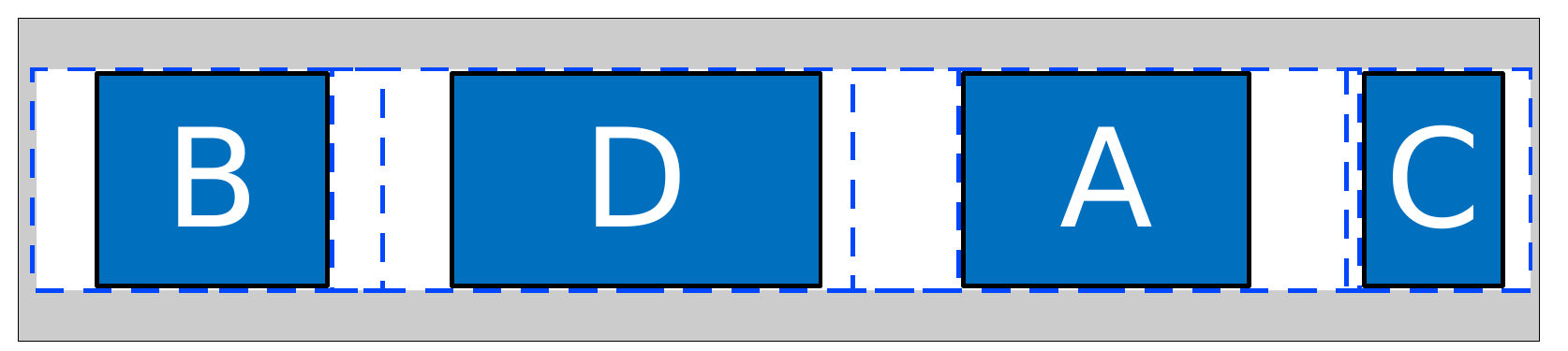}} 
  \caption{Greedy based Single Row Ordering.
  (a) At first all candidates are sorted by blank space.
  (c) One possible ordering solution where each candidate chooses the right end position.
  (e) Another possible ordering solution.
  }
  \label{fig:ordering}
  \vspace{-.1in}
\end{figure}

One example of the greedy approach is illustrated in Fig. \ref{fig:ordering},
where four character candidates $A$, $B$, $C$ and $D$ are to be ordered.
In Fig. \ref{fig:ordering}(a), they are sorted decreasingly by blank space. 
Then all the candidates are inserted one by one.
From the second candidate, each insertion has two options: left side or right side of the whole packed candidates.
For example, if $A$ is inserted at the right of $D$, $B$ has two insertion options:
one is at the right side of $A$ (Fig. \ref{fig:ordering}(b)), another is at the left side of $A$ (Fig. \ref{fig:ordering}(d)).
Given different choices of candidate $B$, Fig. \ref{fig:ordering}(c) and Fig. \ref{fig:ordering}(e) give corresponding final solutions.
Since from the second candidate each one has two choices, by this greedy approach $n$ candidates will generate $2^{n-1}$ possible solutions. 

For the asymmetrical cases, the optimality does not hold anymore.
To compensate the losing, E-BLOW consists of a refinement stage.
For $n$ characters \{$c_1,\dots,c_n$\}, single row ordering can have $n!$ possible solutions.
We avoid enumerating such huge solutions, and take advantage of the order in symmetrical blank assumption.
That is, we pick up one best solution from the $2^{n-1}$ possible ones.
Noted that although considering $2^{n-1}$ instead of $n!$ options cannot guarantee optimal single row packing,
our preliminary results show that the solution quality loss is negligible in practice.

The refinement is based on dynamic programming, and the details are shown in Algorithm \ref{alg:refine}.
Refine(k) generates all possible order solutions for the first $k$ characters \{$c_1, \dots, c_k$\}.
Each order solution is represented as a set $(w,l,r,O)$, where $w$ is the total length of the order, $l$ is the left blank of the left character,
$r$ is the right blank of the right character, and $O$ is the character order.
At the beginning, an empty solution set $S$ is initialized (line 1).
If $k = 1$, then an initial solution $(w_1, sl_1, sr_1, \{c_1\})$ would be generated (line 2).
Here $w_1, sl_1$, and $sr_1$ are width of first character $c_1$, left blank of $c_1$, and right blank of $c_1$.
If $k > 1$, then \textit{Refine}(k) will recursively call \textit{Refine}(k-1) to generate all old partial solutions.
All these partial solutions will be updated by adding candidate $c_k$ (lines 5-9).

\begin{algorithm} [htb]
\caption{Refine(k)}
\label{alg:refine}
\begin{algorithmic}[1]
  \Require{k characters \{$c_1,\dots,c_k$\};}
  \If { k = 1 }
      \State Add $(w_1, sl_1, sr_1, \{c_1\})$ into $S$;
  \Else
      \State Refine(k-1);
      \For{ each partial solution $(w, l, r, O)$}
          \State Remove $(w, l, r, O)$ from $S$;
          \State Add $(w+w_k-\textrm{min}(sr_k, l), sl_k, r, \{c_k, O\})$ into $S$;
          \State Add $(w+w_k-\textrm{min}(sl_k, r), l, sr_k, \{O, c_k\})$ into $S$;
      \EndFor
      \If {size of $S$ $\ge$ $threshold$}
          \State Prune inferior solutions in $S$;
      \EndIf
  \EndIf
\end{algorithmic}
\end{algorithm}

We propose pruning techniques to speed-up the dynamic programming process.
Let us introduce the concept of inferior solutions.
For any two solutions $S_A = (w_a, l_a, r_a, O_a)$ and $S_B = (w_b, l_b, r_b, O_b)$,
we say $S_B$ is \textbf{inferior} to $S_A$ if and only if $w_a \ge w_b$, $l_a \le l_b$ and $r_a \le r_b$.
Those inferior solutions would be pruned during pruning section (lines 10-12).
In our implementation, the $threshold$ is set to 20.


\subsection{Post-Swap and Post-Insertion}
\label{sec:post}

After refinement, a post-swap stage is applied to further improve the performance.
In each swap operation, an unselected character would be swapped with a character on stencil, if such swap can improve the writing time.
The post-swap is implemented using a greedy flavor that consists of two steps.
First, all the unselected characters are sorted.
Second, the unselected characters would try to swap with the characters on stencils one by one.

After post-swap, a post-insertion stage is applied to further insert more characters into stencil.
Different from the greedy insertion approach in \cite{EBL_TCAD2012_Yuan} that new characters can be only inserted into one row's right end.
We consider to insert characters into the middle part of rows.
Generally speaking, the character with higher profit value (\ref{eq:profit}) would have a higher priority to be inserted into rows.
We propose a character insertion algorithm to insert some additional characters into the rows.
The insertion is formulated as a maximum weighted matching problem \cite{FLOW_CSUR86_Galil},
under the constraint that for each row there is at most one character can be inserted.
Although this assumption may loss some optimality, in practical it works quite well as usually the remaining space for a row is very limited.

\begin{figure}[tb!]
  \centering
  \subfigure[]{\includegraphics[width=0.30\textwidth]{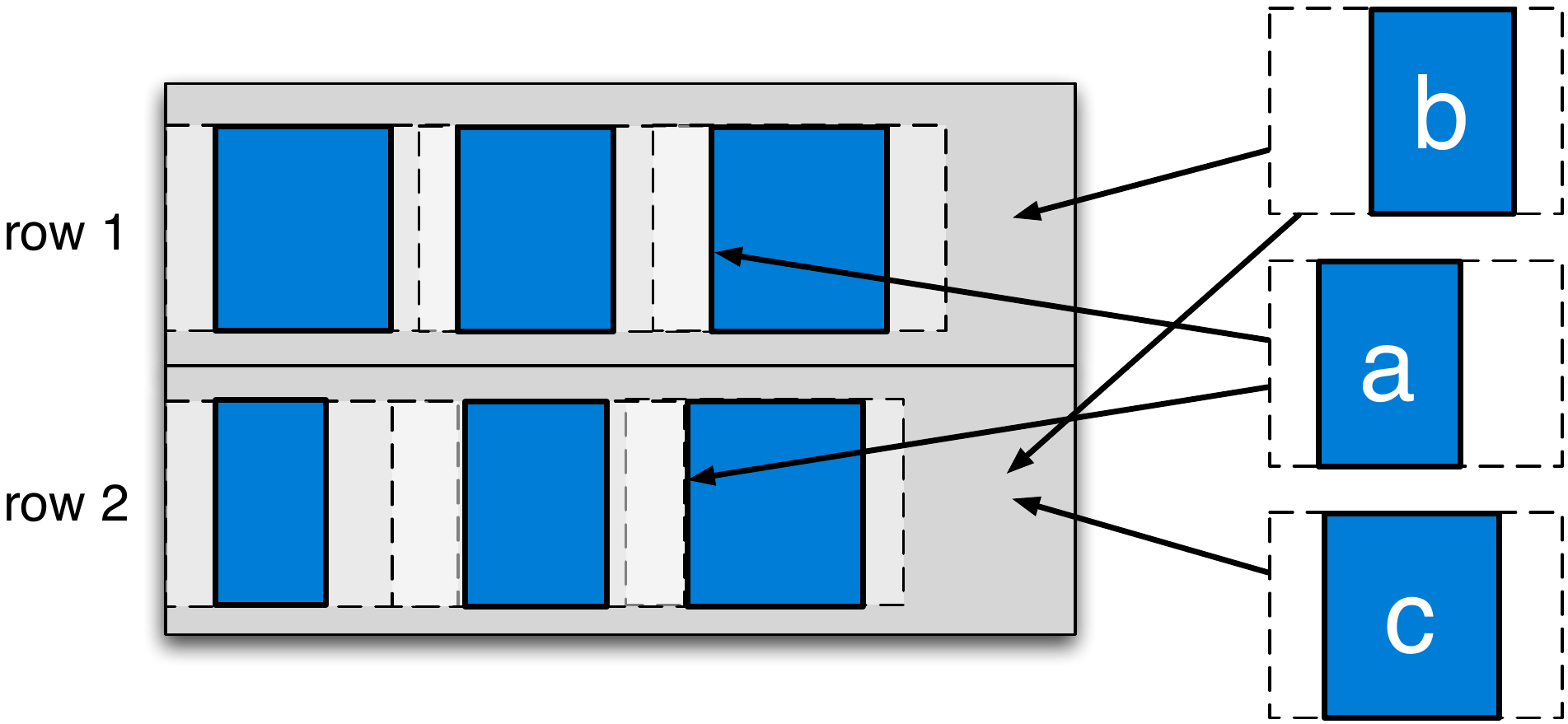}}
  \subfigure[]{\includegraphics[width=0.14\textwidth]{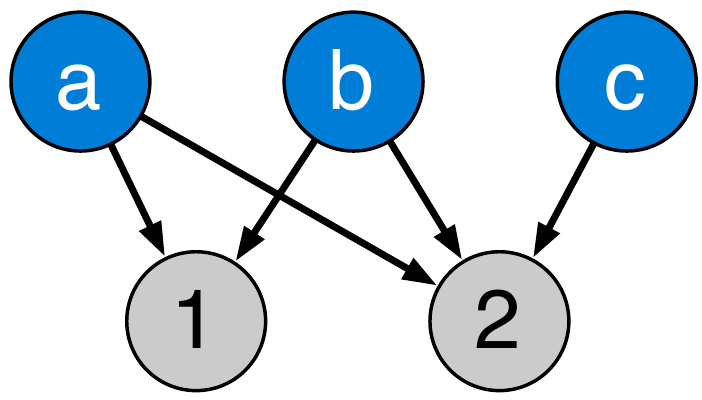}}
  \caption{Example of maximum weighted matching based post character insertion.
  (a) Three additional characters $a, b, c$ and two rows.
  (b) Corresponding bipartite graph to represent the relationships among characters and rows.
  }
  \label{fig:post}
\end{figure}

Fig. \ref{fig:post} illustrates one example of the character insertion.
As shown in Fig. \ref{fig:post} (a), there are two rows (row 1, row 2) and three additional characters ($a, b, c$).
Characters $a$ and $b$ can be inserted into either row 1 or row 2, but character $c$ can only be inserted into row 2.
It shall be noted that the insertion position is labeled by arrows.
For example, two arrows from character $a$ mean that $a$ can be inserted into the middle of each row.
We build up a bipartite graph to represent the relationships among characters and rows (see Fig. \ref{fig:post} (b)).
Each edge is associated with a cost as character's profit.
By utilizing the bipartite graph, the best character insertion can be solved by finding a maximum weighted matching.

Given $n$ additional characters, we search the possible insertion positions under each row.
The time complexity of searching all the possibilities is $O(nmC)$,
where $m$ is the total row number and $C$ is the maximum character number on each row.
We propose two heuristics to speed-up the search process.
First, to reduce $n$, we only consider those additional characters with high profits.
Second, to reduce $m$, we skip those rows with very little empty space.

\vspace{.1in}
\section{E-BLOW for $\mathsf{2DOSP}$}
\label{sec:2d}

Now we consider a more general case: the blank spaces of characters are non-uniform along both horizontal and vertical directions. 
This problem is referred to $\mathsf{2DOSP}$ problem.
In \cite{EBL_TCAD2012_Yuan} the $\mathsf{2DOSP}$ problem was transformed into a floorplanning problem.
However, several key differences between traditional floorplanning and $\mathsf{OSP}$ were ignored.
(1) In $\mathsf{OSP}$ there is no wirelength to be considered, while at floorplanning wirelength is a major optimization objective.
(2) Compared with complex IP cores, lots of characters may have similar sizes.
(3) Traditional floorplanner could not handle the problem size of modern MCC design.

\subsection{ILP Formulation}
\label{sec:2d-ilp}

\begin{table}[bt]
 \centering 
 \caption{Notations used in 2D-ILP Formulation}
 \label{tab:notations2}
\resizebox{8.4cm}{!} {
 \begin{tabular}{|l|l|}
 \hline
 \hline
 $W (H)$                  & width (height) constraint of stencil \\
 \hline
 $w_i (h_i)$              & width (height) of candidate $c_i$ \\
 \hline
 $o_{ij}^h (o_{ij}^v)$    & horizontal (vertical) overlap between $c_i$ and $c_j$ \\
 \hline
 $w_{ij} (h_{ij})$        & $w_{ij} = w_i - o_{ij}^h$, $h_{ij} = h_i - o_{ij}^v$\\
 \hline
 $a_i$                    & 0-1 variable, $a_i = 1$ if $c_i$ is on stencil\\
\hline
\hline
\end{tabular}
}
\end{table}

Here we will show that $\mathsf{2DOSP}$ can be formulated as integer linear programming (ILP) as well.
Compared with $\mathsf{1DOSP}$, $\mathsf{2DOSP}$ is more general: the blank spaces of characters are non-uniform along both horizontal and vertical directions. 
The $\mathsf{2DOSP}$ problem can be also formulated as an ILP formulation (\ref{eq:2ilp}).
For convenience, Table \ref{tab:notations2} lists some notations used in the ILP formulation.
The formulation is motivated by \cite{FLOOR_TCAD91_Sutanthavibul}, but the difference is that our formulation can optimize both placement constraints and character selection, simultaneously.
{
\small
\begin{align}
  \textrm{min} \ & T_{total}  		 & \label{eq:2ilp}\\
  \textrm{s.t.}\
  & T_{total} \ge T_{c}^{VSB} - \sum_{i=1}^{n} R_{ic} \cdot a_i        & \forall c \in P    \label{2ilp_a}\tag{\ref{eq:2ilp}$a$}\\
  &   x_i + w_{ij} \le x_j + W (2+p_{ij}+q_{ij}-a_i-a_j)	           & \forall i,j        \label{2ilp_b}\tag{\ref{eq:2ilp}$b$}\\
  &   x_i - w_{ji} \ge x_j - W (3+p_{ij}-q_{ij}-a_i-a_j)	           & \forall i,j        \label{2ilp_c}\tag{\ref{eq:2ilp}$c$}\\
  &   y_i + h_{ij} \le y_j + H (3-p_{ij}+q_{ij}-a_i-a_j)	           & \forall i,j        \label{2ilp_d}\tag{\ref{eq:2ilp}$d$}\\
  &   y_i - h_{ji} \ge y_j - H (4-p_{ij}-q_{ij}-a_i-a_j)	           & \forall i,j        \label{2ilp_e}\tag{\ref{eq:2ilp}$e$}\\
  &   0 \le x_i + w_i \le W, \ \  0 \le y_i + h_i \le H				   & \forall i          \label{2ilp_f}\tag{\ref{eq:2ilp}$f$}\\
  &   p_{ij}, q_{ij}, a_i: \textrm{0-1 variable} 	                   & \forall i,j        \label{2ilp_g}\tag{\ref{eq:2ilp}$g$}
\end{align}
}
where $a_i$ indicates whether candidate $c_i$ is on the stencil, 
$p_{ij}$ and $q_{ij}$ represent the location relationships between $c_i$ and $c_j$. 
The number of variables is $O(n^2)$, where $n$ is number of characters.
We can see that if $a_i = 0$, constraints (\ref{2ilp_b}) - (\ref{2ilp_e}) are not active.
Besides, it is easy to see that when $a_i = a_j = 1$, for each of the four possible choices of $(p_{ij}, q_{ij}) = (0,0), (0,1), (1,0), (1,1)$, only one of the four inequalities (\ref{2ilp_b}) - (\ref{2ilp_e}) are active.
For example, with ($a_i, a_j, p_{ij}, q_{ij}$) = (1,1,1,1), only the constraint (\ref{2ilp_e}) applies,
which allows character $c_i$ to be anywhere above character $c_j$.
The other three constraints (\ref{2ilp_b})-(\ref{2ilp_d}) are always satisfied for any permitted values of ($x_i, y_i$) and ($x_j, y_j$).

Program (\ref{eq:2ilp}) can be relaxed to linear programming (LP) by replacing constraint (\ref{2ilp_g}) as:
\begin{displaymath}
0 \le p_{ij}, q_{ij}, a_i \le 1
\end{displaymath}
However, similar to the discussion in $\mathsf{1DOSP}$, the relaxed LP solution provides no information or guideline to the packing, i.e., every $a_i$ is set as $1$, and every $p_{ij}$ is set as $0.5$.
In other words, this LP relaxation provides no useful information to guide future rounding:
all the character candidates are selected and no ordering relationship is determined.
Therefore we can see that LP rounding method cannot be effectively applied to program (\ref{eq:2ilp}).



\subsection{Clustering based Simulated Annealing}

\begin{figure} [tb]
  \centering
  \includegraphics[width=0.34\textwidth]{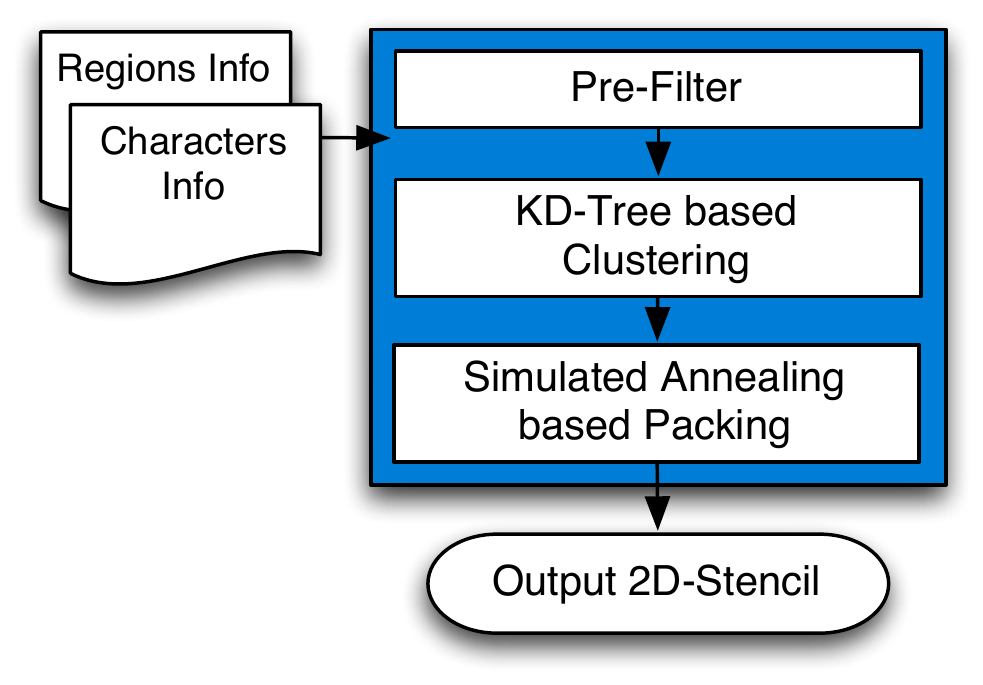}
  \caption{E-BLOW overall flow for $\mathsf{2DOSP}$.}
  \label{fig:Flow2D}
\end{figure}

To deal with all these limitations of ILP formulation, an fast packing framework is proposed (see Fig. \ref{fig:Flow2D}).
Given the input character candidates, the pre-filter process is first applied to remove characters with bad profit (defined in (\ref{eq:profit})).
Then the second step is a clustering algorithm to effectively speed-up the design process.
Followed by the final floorplanner to pack all candidates.

Clustering is a well studied problem, and there are many of works and applications in VLSI
\cite{PAR_JVLSI95_Alpert}
However, previous methodologies cannot be directly applied here.
First, traditional clustering is based on netlist, which provides the all clustering options.
Generally speaking, netlist is sparse, but in $\mathsf{OSP}$ the connection relationships are so complex that any two characters can be clustered, and totally there are $O(n^2)$ clustering options.
Second, given two candidates $c_i$ and $c_j$, there are several clustering options.
For example, horizontal clustering and vertical clustering may have different overlapping blank space.

The main ideas of our clustering are iteratively search and group each character pair ($c_i, c_j$) with similar blank spaces, profits, and sizes.
Character $c_i$ is said to be similar to $c_j$, if the following condition is satisfied:
\begin{equation}
  \left\{
  \begin{array}{c}
  \textrm{max} \{ {|w_i-w_j|}/{w_j},    {|h_i-h_j|}/{h_j}       \} \le bound \\
  \textrm{max} \{ {|sh_i-sh_j|}/{sh_j}, {|sv_i - sv_j|}/{sv_j}  \} \le bound \\
  {|profit_i-profit_j|}/{profit_j}                                 \le bound
  \end{array}
  \right.
\end{equation}
where $w_i$ and $h_i$ are the width and height of $c_i$.
$sh_i$ and $sv_i$ are the horizontal blank space and vertical blank space of $c_i$, respectively.
In our implementation, $bound$ is set as 0.2.
We can see that in clustering, all the size, blanks, and profits are considered.

The details of our clustering procedure are shown in Algorithm \ref{alg:cluster}.
First all the initial character candidates are sorted by $profit_i$ (line 2),
so those characters with more shot number reduction are tend to be clustered.
Then all characters are labeled as \textit{unclustered} (line 3).
The clustering (lines 3-10) is repeated until no characters can be further merged.
When cluster $c_i, c_j$, the information of $c_i$ is modified to incorporate $c_j$, and the $c_j$ is labeled as \textit{clustered}.

\begin{algorithm}[htb]
\caption{KD-Tree based Clustering}
\label{alg:cluster}
\begin{algorithmic}[1]
  \Require{set of character candidates.}
  \State  {Sort all candidates by $profit_i$;}
  \State  {Set each candidates $c_i$ to \textit{unclustered};}
  \Repeat
      \ForAll {unclustered candidate $c_i$}
          \If {can find similar unclustered character $c_j$}
              \State Update information of $c_i$ to incorporate $c_j$;
              \State Label $c_j$ as \textit{clustered};
          \EndIf
      \EndFor
  \Until {no character can be merged}
\end{algorithmic}
\end{algorithm}

For each candidate $c_i$, finding available $c_j$ may need $O(n)$, and complexity of the horizontal clustering and vertical clustering are both $O(n^2)$.
Then the complexity of the whole procedure is $O(n^2)$, where $n$ is the number of candidates.

A KD-Tree~\cite{KDTree_CACM75_Bentley} is used to speed-up the process of finding available pair $(c_i, c_j)$.
It provides fast $O(log n)$ region searching operations which keeping the time for insertion and deletion small:
insertion, $O(log n)$; deletion of the root, $O(n(k-1)/k)$; deletion of a random node, $O(log n)$.
Using KD-Tree, the complexity of the Algorithm \ref{alg:cluster} can be reduced to $O(nlog n)$.
For instance, given nine character candidates \{$c_1, \dots, c_9$\} as in Fig. \ref{fig:kdtree} (a),
the corresponding KD-Tree is shown in Fig. \ref{fig:kdtree} (b).
Note that KD-Tree can store multiple dimensional vertices, thus a single tree is enough to store all the information regarding width, height, blank spaces, and profits.
For the sake of convenience, here characters are distributed only based on horizontal and vertical blank spaces.
Thus only two dimensional space is illustrated in Fig. \ref{fig:kdtree} (a).
To search candidates with similar blank space with $c_2$ (see the shaded region of Fig. \ref{fig:kdtree} (a)),
it may need $O(n)$ time to scan all candidates, where $n$ is the total candidate number.
However, under the KD-Tree structure, this search procedure can be resolved in $O(log n)$.
All candidates scanned ($c_1-c_5$) are illustrated in Fig. \ref{fig:kdtree} (b).
Particularly, after scanning the $c_5$, since $c_5$ is out of the search range, we can make sure the whole sub-tree rooted by $c_7$ is out of the search range as well.

\begin{figure} [tb]
  \centering
  \subfigure[]{\includegraphics[width=0.20\textwidth]{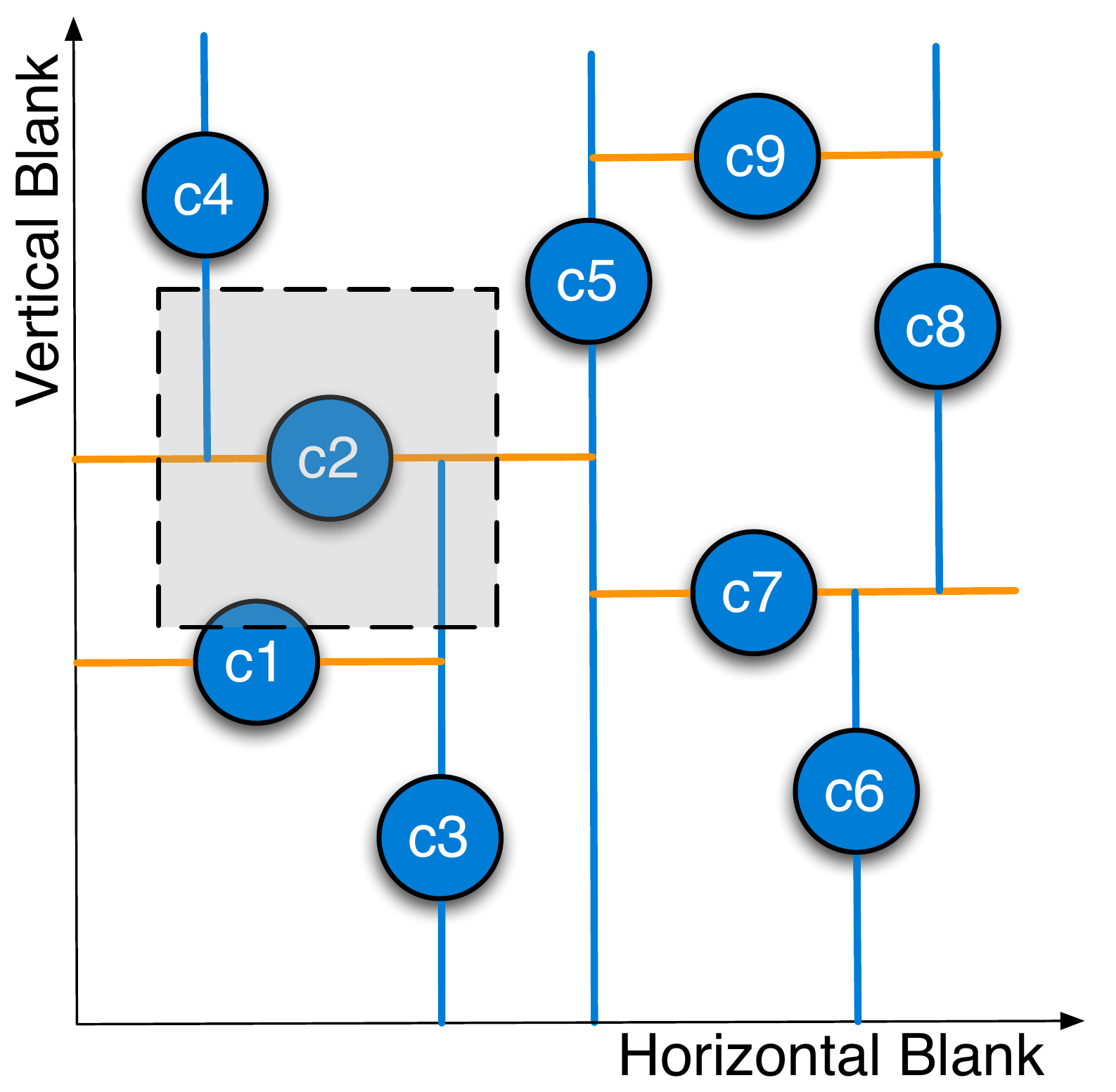}}
  \subfigure[]{\includegraphics[width=0.24\textwidth]{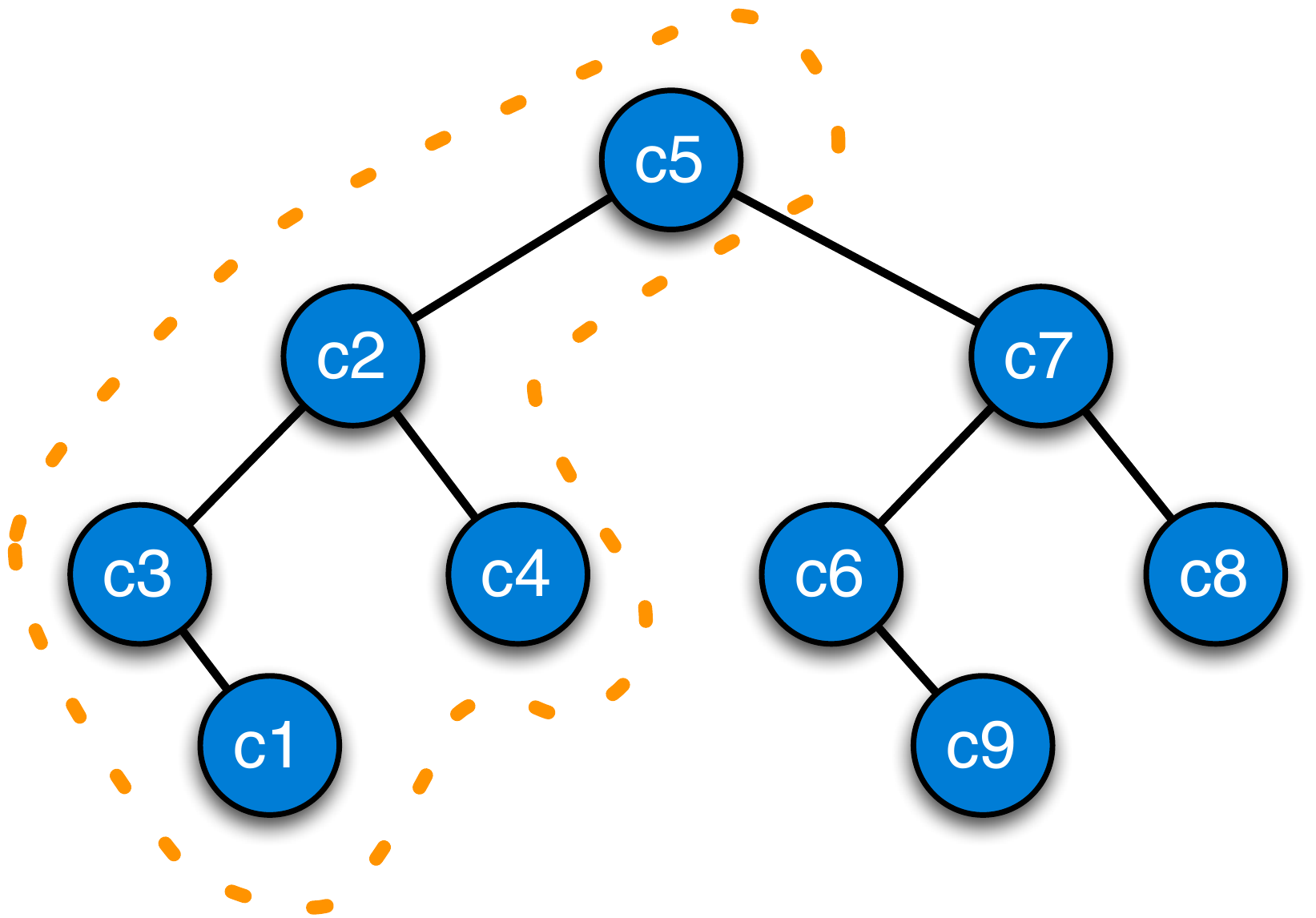}}
  \vspace{-.1in}
  \caption{KD-Tree based region searching. (a) A two dimensional space split by eight points; (b) The corresponding two dimensional KD-Tree.}
  \label{fig:kdtree}
\end{figure}

In \cite{EBL_TCAD2012_Yuan}, the $\mathsf{2DOSP}$ is transformed into a fixed-outline floorplanning problem.
If a character candidate is outside the fixed-outline, then the character would not be prepared on stencil.
Otherwise, the character candidate would be selected and packed on stencil.
Parquet \cite{FLOOR_TVLSI03_Adya} was adopted as simulated annealing engine,
and Sequence Pair \cite{FLOOR_TCAD96_SP} was used as a topology representation.
In E-BLOW we apply a simulated annealing based framework similar to that in \cite{EBL_TCAD2012_Yuan}.
To demonstrate the effectiveness of our pre-filter and clustering methodologies, E-BLOW uses the same parameters.

\section{Experimental Results}
\label{sec:result}

\begin{table*}[tb]
\centering 
\caption{Result Comparison for $\mathsf{1DOSP}$}
\label{table:1d}
\resizebox{18.4cm}{!} {
\resizebox{17.84cm}{!} {
\begin{tabular}{|r|r|r|r|r|r|r|r|r|r|r|r|r|r|r|}
\hline 
\hline 
 &char &CP & \multicolumn{3}{c|}{Greedy in \cite{EBL_TCAD2012_Yuan}} & \multicolumn{3}{c|}{\cite{EBL_TCAD2012_Yuan}}
           & \multicolumn{3}{c|}{\cite{EBL_ISPD2014_Kuang}}          & \multicolumn{3}{c|}{E-BLOW}\\
 \cline{4-15} & \# & \#  &T &char\# &CPU(s)  &T &char\# &CPU(s)        &T &char\# &CPU(s)  &T &char\# &CPU(s)\\
 \hline
 1D-1 &1000 &1 	  &64891  &912   &0.1   &50809     &926     &13.5      &\textbf{19095}   &940   &0.005    &       {19479}    &940    &2.1    \\
 1D-2 &1000 &1 	  &99381  &884   &0.1   &93465     &854     &11.8      &35295            &864   &0.005    &\textbf{34974}    &866    &1.7    \\
 1D-3 &1000 &1 	  &165480 &748   &0.1   &152376    &749     &9.13      &69301            &757   &0.005    &       {67209}    &766    &1.7    \\
 1D-4 &1000 &1 	  &193881 &691   &0.1   &193494    &687     &7.7       &\textbf{92523}   &703   &0.005    &\textbf{93816}    &703    &4.5    \\
 1M-1 &1000 &10   &63811  &912   &0.1   &53333     &926     &13.5      &39026            &938   &0.01     &\textbf{37848}    &944    &3.8    \\
 1M-2 &1000 &10   &104877 &884   &0.1   &95963     &854     &11.8      &77997            &864   &0.01     &\textbf{75303}    &874    &3.5    \\
 1M-3 &1000 &10   &172834 &748   &0.1   &156700    &749     &9.2       &138256           &758   &0.56     &\textbf{132773}   &774    &9.3    \\
 1M-4 &1000 &10   &200498 &691   &0.1   &196686    &687     &7.7       &176228           &698   &0.36     &\textbf{173193}   &711    &7.4    \\
 1M-5 &4000 &10   &274992 &3604  &1.0   &255208    &3629    &1477.3    &204114           &3660  &0.03     &\textbf{202401}   &3680   &37.9   \\
 1M-6 &4000 &10   &437088 &3341  &1.0   &417456    &3346    &1182      &357829           &3382  &0.03     &\textbf{348007}   &3420   &48.4   \\
 1M-7 &4000 &10   &650419 &3000  &1.0   &644288    &2986    &876       &568339           &3016  &0.59     &\textbf{563054}   &3064   &54.0   \\
 1M-8 &4000 &10   &820013 &2756  &1.0   &809721    &2734    &730.7     &731483           &2760  &0.42     &\textbf{721149}   &2818   &54.7   \\
 \hline
 Avg. & - & - 	  &270680.4&1597.6&0.4  &259958.3  &1594.0  &362.5     &209123.8&1611.7&0.17              &205767.2          &1630.7 &16.6  \\
 Ratio &-&-
 &\textbf{1.32}  &0.98   &0.02
 &\textbf{1.26}  &0.98   &19.01
 &\textbf{1.02}  &0.99   &0.01
 &\textbf{1.0}   &1.0    &1.0     \\
\hline \hline
\end{tabular}
}
}
\end{table*}

E-BLOW is implemented in C++ programming language and executed on a Linux machine with two 3.0GHz CPU and 32GB Memory.
GUROBI~\cite{TOOL_gurobi} is used to solve ILP/LP. 
The benchmark suite from \cite{EBL_TCAD2012_Yuan} are tested (1D-1, $\dots$, 1D-4, 2D-1, $\dots$, 2D-4).
To evaluate the algorithms for MCC system,
eight benchmarks (1M-x) are generated for $\mathsf{1DOSP}$ and the other eight (2M-x) are generated for the $\mathsf{2DOSP}$ problem.
In these new benchmarks, character projection (CP) number are all set to 10.
For each small case (1M-1, $\dots$, 1M-4, 2M-1, $\dots$, 2M-4) the character candidate number is 1000,
and the stencil size is set to $1000\mu m \times 1000 \mu m$.
For each larger case (1M-5 , $\dots$, 1M-8, 2M-5, $\dots$, 2M-8) the character candidate number is 4000,
and the stencil size is set to $2000\mu m \times 2000 \mu m$.
The size and the blank width of each character are similar to those in \cite{EBL_TCAD2012_Yuan}.
It shall be noted that \cite{EBL_TCAD2012_Yuan} is aimed for single CP system,
for MCC system it is modified to optimize the total writing time of all the regions.

\subsection{Comparison for $\mathsf{1DOSP}$}

For $\mathsf{1DOSP}$, Table \ref{table:1d} compares E-BLOW with the greedy method in \cite{EBL_TCAD2012_Yuan}, the heuristic framework in \cite{EBL_TCAD2012_Yuan}, and the algorithms in \cite{EBL_ISPD2014_Kuang}.
We have obtained the programs of \cite{EBL_TCAD2012_Yuan} and executed them in our machine.
The results of \cite{EBL_ISPD2014_Kuang} are directly from their paper.
Column  ``char \#'' is number of character candidates, and column ``CP\#'' is number of character projections.
For each algorithm, we report \textbf{``T'', ``char\#'' and ``CPU(s)''},
where ``T'' is the writing time of the E-Beam system, ``char\#'' is the character number on final stencil, and ``CPU(s)'' reports the runtime.
From Table \ref{table:1d} we can see E-BLOW achieves better performance than both greedy method and heuristic method in \cite{EBL_TCAD2012_Yuan}.
Compared with E-BLOW, the greedy method has $32\%$ more system writing time, while \cite{EBL_TCAD2012_Yuan} introduces $27\%$ more system writing time.
One possible reason is that different from the greedy/heuristic methods, E-BLOW proposes mathematical formulations to provide global view.
Additionally, due to the successive rounding scheme, E-BLOW is around 22$\times$ faster than the work in \cite{EBL_TCAD2012_Yuan}.

E-BLOW is further compared with one recent $\mathsf{1DOSP}$ solver \cite{EBL_ISPD2014_Kuang} in Table \ref{table:1d}.
E-BLOW found stencil placements with best E-Beam system writing time for 10 out of 12 test cases.
In addition, for all the MCC system cases (1M-1, $\dots$, 1M-8) E-BLOW outperforms \cite{EBL_ISPD2014_Kuang}.
One possible reason is that to optimize the overall throughput of the MCC system, a global view is necessary to balance the throughputs among different regions.
E-BLOW utilizes the mathematical formulations to provide such global optimization.
Although the linear programming solvers are more expensive than the deterministic heuristics in \cite{EBL_ISPD2014_Kuang},
the runtime of E-BLOW is reasonable that each case can be finished in 20 seconds on average.

We further demonstrate the effectiveness of the fast ILP convergence (Section \ref{sec:ilp}) and post-insertion (Section \ref{sec:post}).
We denote \textbf{E-BLOW-0} as E-BLOW without these two techniques, and denote \textbf{E-BLOW-1} as E-BLOW with these techniques.
Fig. \ref{fig:post_shot} and Fig. \ref{fig:post_runtime} compare E-BLOW-0 and E-BLOW-1, 
in terms of system writing time and runtime, respectively.
From Fig. \ref{fig:post_shot} we can see that applying fast ILP convergence and post-insertion can effectively E-Beam system throughput,
that is, averagely 9\% system writing time reduction can be achieved.
In addition, Fig. \ref{fig:post_runtime} demonstrates the performance of the fast ILP convergence (see Section \ref{sec:ilp}).
We can see that in 11 out of 12 test cases, the fast ILP convergence can effectively reduce E-BLOW CPU time.
The possible reason for the slow down in case 1D-4 is that when fast convergence is called, if there are still many unsolved $a_{ij}$ variables, ILP solver may suffer from runtime overhead problem.
However, if more successive rounding iterations are applied before ILP convergence, less runtime can be reported.

\begin{figure}[htb]
  \centering
  \includegraphics[width=0.44\textwidth]{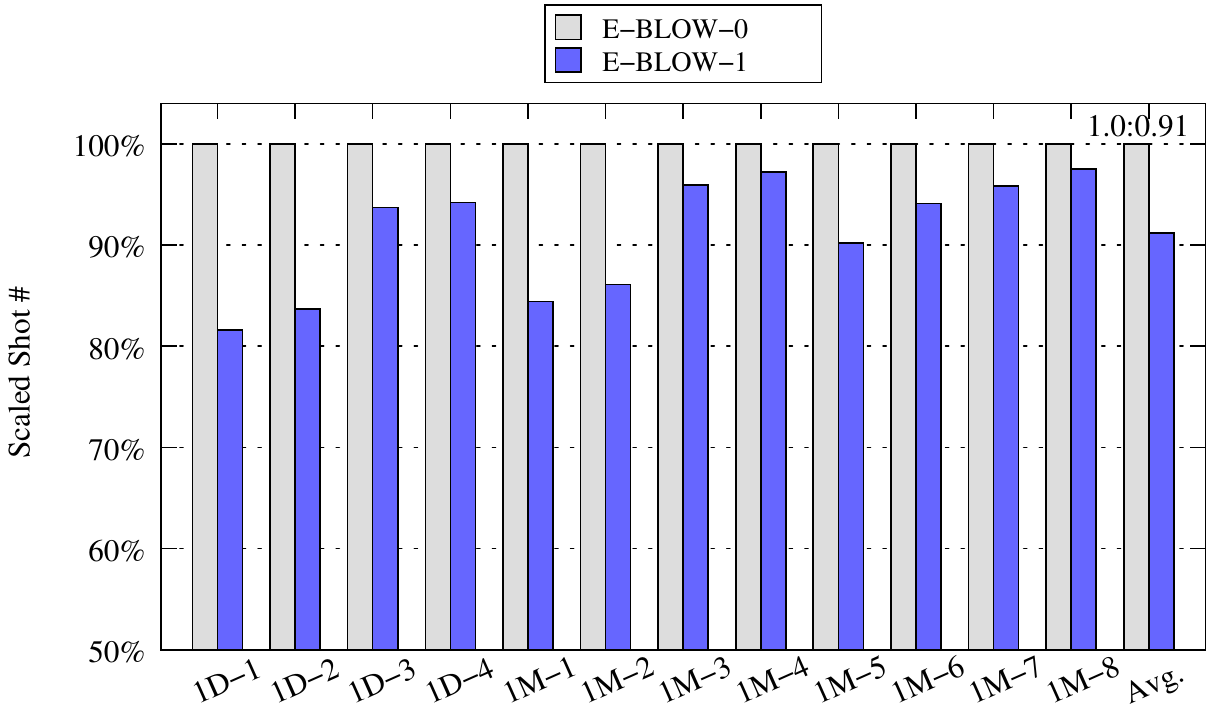}
  \caption{The comparison of E-Beam system writing times between E-BLOW-0 and E-BLOW-1.}
  \label{fig:post_shot}
  \vspace{-.2in}
\end{figure}

\begin{figure}[htb]
  \centering
  \includegraphics[width=0.44\textwidth]{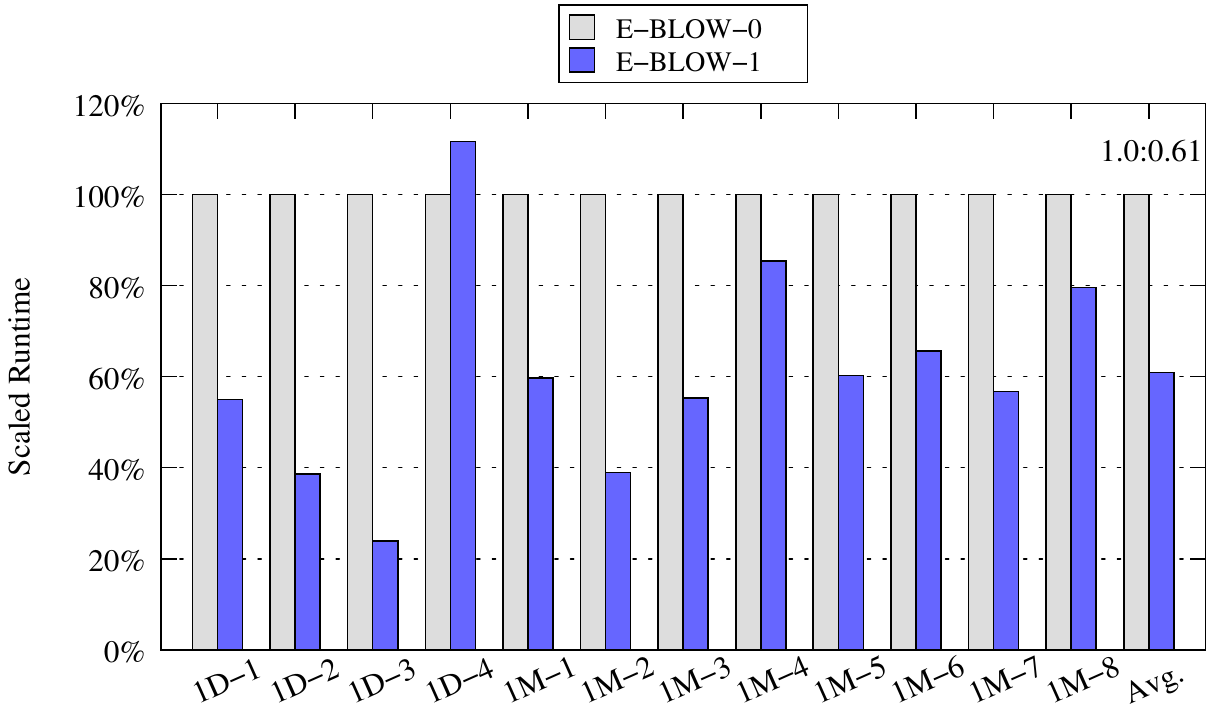}
  \caption{The comparison of runtime between E-BLOW-0 and E-BLOW-1.}
  \label{fig:post_runtime}
  \vspace{-.2in}
\end{figure}

\begin{table*}[tb]
\centering 
\caption{Result Comparison for $\mathsf{2DOSP}$}
\label{table:2d}
\resizebox{14.4cm}{!} {
\begin{tabular}{|r|r|r|r|r|r|r|r|r|r|r|r|}
\hline 
\hline 
 &char &CP & \multicolumn{3}{c|}{Greedy in \cite{EBL_TCAD2012_Yuan}} & \multicolumn{3}{c|}{\cite{EBL_TCAD2012_Yuan}}&\multicolumn{3}{c|}{E-BLOW}\\
 \cline{4-12} & \# & \#  &T & char \# & CPU(s) &T & char \# & CPU(s) &T & char \# & CPU(s)\\
 \hline
 2D-1 &1000 &1    &159654  &734  &2.1  	  &107876  &826  &329.6 	&105723  &789  &65.5 \\
 2D-2 &1000 &1    &269940  &576  &2.4 	  &166524  &741  &278.1 	&170934  &657  &52.5 \\
 2D-3 &1000 &1    &290068  &551  &2.6  	  &210496  &686  &296.7 	&178777  &663  &56.4 \\
 2D-4 &1000 &1    &327890  &499  &2.7  	  &240971  &632  &301.7 	&179981  &605  &54.7 \\
 2M-1 &1000 &1    &168279  &734  &2.1  	  &122017  &811  &313.7 	&91193   &777  &58.6 \\
 2M-2 &1000 &1    &283702  &576  &2.4 	  &187235  &728  &286.1 	&163327  &661  &48.7 \\
 2M-3 &1000 &1    &298813  &551  &2.6  	  &235788  &653  &289.0   &162648  &659  &52.3 \\
 2M-4 &1000 &1    &338610  &499  &2.7  	  &270384  &605  &285.6 	&195469  &590  &53.3 \\
 2M-5 &4000 &10   &824060  &2704 &19.0	  &700414  &2913 &3891.0  &687287  &2853 &59.0 \\
 2M-6 &4000 &10   &1044161 &2388 &20.2	  &898530  &2624 &4245.0  &717236  &2721 &60.7 \\
 2M-7 &4000 &10   &1264748 &2101 &21.9 	  &1064789 &2410 &3925.5	&921867  &2409 &57.1 \\
 2M-8 &4000 &10   &1331457 &2011 &22.8 	  &1176700 &2259 &4550.0	&1104724 &2119 &57.7 \\
 \hline
 Avg. & -   & -   &550115  &1218.1&8.3    &448477  &1324 &1582.7    &389930.5&1291.9&56.375\\
 Ratio&-    & -   &\textbf{1.41}    &0.94  &\textbf{0.15}   &\textbf{1.15}    &1.02 &\textbf{28.1}      &\textbf{1.0}       &1.0     &\textbf{1.0}   \\
\hline \hline
\end{tabular}
}
\end{table*}

\subsection{Comparison for $\mathsf{2DOSP}$}
For $\mathsf{2DOSP}$, Table \ref{table:2d} gives the similar comparison.
For each algorithm, we also record ``T'', ``char \#'' and ``CPU(s)'',
where the meanings are the same with that in Table \ref{table:1d}.
Compared with E-BLOW, although the greedy algorithm is faster, its design results  would introduce 41\% more system writing time.
Furthermore, compared with E-BLOW, although the framework in \cite{EBL_TCAD2012_Yuan} puts 2\% characters onto stencil,
it gets 15\% more system writing time.
The possible reason is that in E-BLOW the characters with similar writing time are clustered together.
The clustering method can help to speed-up the packaging, so E-BLOW is $28\times$ faster than \cite{EBL_TCAD2012_Yuan}.
In addition, after clustering the character number can be reduced.
With smaller solution space, the simulated annealing engine is easier to achieve a better solution, in terms of system writing time.

From both tables we can see that compared with \cite{EBL_TCAD2012_Yuan}, E-BLOW can achieve a better tradeoff between runtime and system throughput.

\begin{table*}[tb]
\centering
\caption{ILP v.s. EBLOW}
\label{table:small}
\resizebox{10.4cm}{!} {
\begin{tabular}{|r|r|r|r|r|r|r|r|r|}
  \hline \hline
  & candidate\# &\multicolumn{4}{c|}{ILP}	& \multicolumn{3}{c|}{E-BLOW}\\
  \cline{3-9}
  &       &binary\#	&T	&char\#	&CPU(s)	  &T	&char\#  &CPU(s)\\
	\hline
  1T-1    &8   &64    &434   &6   &0.5        &434   &6 &0.1  \\
  1T-2    &10  &100   &1034  &6   &26.1       &1034  &6 &0.2  \\
  1T-3    &11  &121   &1222  &6   &58.3       &1222  &6 &0.2  \\
  1T-4    &12  &144   &1862  &6   &1510.4     &1862  &6 &0.2  \\
  1T-5    &14  &196   &NA    &NA  &$>$3600    &2758  &6 &0.1  \\
	\hline
  2T-1    &6   &66    &60    &6   &37.3       &207   &5 &0.1  \\
  2T-2    &8   &120   &354   &6   &40.2       &653   &7 &0.1  \\
  2T-3    &10  &190   &1050  &6   &436.8      &4057  &4 &0.1  \\
  2T-4    &12  &276   &NA    &NA  &$>$3600    &4208  &5 &0.2  \\
	\hline \hline
\end{tabular}
}
\end{table*}

\subsection{E-BLOW vs. ILP}
\label{sec:vs-ilp}

We further compare the E-BLOW with the ILP formulations (\ref{eq:1ilp}) and (\ref{eq:2ilp}).
Although for both $\mathsf{OSP}$ problems the ILP formulations can find optimal solutions theoretically, they may suffer from runtime overhead.
Therefore, we randomly generate nine small benchmarks, five for $\mathsf{1DOSP}$ (``1T-x'') and four for $\mathsf{2DOSP}$ (``2T-x'').
The sizes of all the character candidates are set to $40\mu m \times 40 \mu m$.
For $\mathsf{1DOSP}$ benchmarks, the row number is set to 1, and the row length is set to 200.
The comparisons are listed in Table \ref{table:small}, where column ``candidate\#'' is the number of character candidates.
``\textbf{ILP}'' and ``\textbf{E-BLOW}'' represent the ILP formulation and our E-BLOW framework, respectively.
In ILP formulation, column ``binary\#'' gives the binary variable number.
For each mode, we report ``T'', ``char\#'' and ``CPU(s)'',
where ``T'' is E-Beam system writing time, ``char\#'' is character number on final stencil, and ``CPU(s)'' is the runtime.
Note that in Table \ref{table:small} the ILP solutions are optimal.

Let us compare E-BLOW with ILP formulation for 1D cases (1T-1, $\dots$, 1T-5).
E-BLOW can achieve the same results with ILP formulations, meanwhile it is very fast that all cases can be finished in 0.2 seconds.
Although ILP formulation can achieve optimal results,
it is very slow that a case with 14 character candidates (1T-5) can not be solved in one hour.
Next, let us compare E-BLOW with ILP formulation for 2D cases (2T-1, $\dots$, 2T-4).
For 2D cases ILP formulations are slow that if the character candidate number is 12, it cannot finish in one hour.
E-BLOW is fast, but with some solution quality penalty.

Although the integral variable number for each case is not huge,
we find that in the ILP formulations, the solutions of corresponding LP relations are vague.
Therefore, expensive search method may cause unacceptable runtimes.
From these cases ILP formulations are impossible to be directly applied in $\mathsf{OSP}$ problem,
as in MCC system character number may be as large as $4000$.

\section{Conclusion}
\label{sec:conclu}

In this paper, we have proposed E-BLOW, a tool to solve OSP problem in MCC system.
For $\mathsf{1DOSP}$, a successive relaxation algorithm and a dynamic programming based refinement are proposed.
For $\mathsf{2DOSP}$, a KD-Tree based clustering method is integrated into simulated annealing framework.
Experimental results show that compared with previous works, E-BLOW can achieve better performance in terms of shot number and runtime, for both MCC system and traditional EBL system.
Note that the extra cost for multiple stencils is mostly the cost of multiple stencil design,
thus different regions tend to have specific stencils to improve the throughput.
However, if a shared stencil is well-designed and optimized that such sharing can achieve very comparable throughput, we can even reduce the stencil design cost.
In that situation, sharing stencil design could be attractive, especially for the companies that have limited design budget.
As EBL, including MCC system, are widely used for mask making and also gaining momentum for direct wafer writing, we believe a lot more research can be done for not only stencil planning, but also EBL aware design.

\section*{Acknowledgment}

This work is supported in part by NSF grants CCF-0644316 and CCF-1218906, SRC task 2414.001, NSFC grant 61128010, and IBM Scholarship.
The authors would like to thank Prof. Shiyan Hu at Michigan Technological University and Zhao Song at University of Texas for helpful comments.

\bibliographystyle{IEEEtran}
\bibliography{./ref/Top,./ref/Algorithm,./ref/Optimization,./ref/EBL,./ref/Lith,./ref/MPL,./ref/Floorplan,./ref/Place,./ref/Software}

\appendix

\noindent
\underline{PROOF OF THEOREM \ref{thm:npc_bss}}

\begin{mylemma}\label{lem:np_bss}
    $\mathsf{BSS}$ problem is in NP.
\end{mylemma}

\begin{proof}
It is easy to see that $\mathsf{BSS}$ problem is in NP.
Given a subset of integer numbers $S' \in S$, we can add them up and verify that their sum is $s$ in polynomial time.
\end{proof}


\begin{mylemma}\label{lem:nph_bss}
$\mathsf{3SAT}~$ $\leq_p$ $\mathsf{BSS}$.
\end{mylemma}

\begin{proof}
In $\mathsf{3SAT}$ problem, we are given $m$ clauses $\{C_1, C_2, \dots, C_m\}$ over $n$ variables $\{y_1, y_2, \dots, y_n\}$.
Besides, there are three literals in each clause, which is the OR of some number of literals.
Eqn. (\ref{eq:ex_3sat}) gives one example of $\mathsf{3SAT}$, where  $n=4$ and  $m=2$.
\begin{equation}\label{eq:ex_3sat}
  (y_1 \vee \bar y_3 \vee \bar y_4) \wedge (\bar y_1 \vee y_2 \vee \bar y_4)
\end{equation}

Without loss of generality, we can have the following assumptions:
\begin{enumerate}
    \item No clause contains both variable $y_i$ and $\bar y_i$.
          Otherwise, any such clause is always true and we can just eliminate them from the formula.
    \item Each variable $y_i$ appears in at least one clause.
          Otherwise, we can just assign any arbitrary value to the variable $y_i$.
\end{enumerate}

To convert a $\mathsf{3SAT}$ instance to a $\mathsf{BSS}$ instance, we create two integer numbers in set $S$ for each variable $y_i$ and three integer numbers in $S$ for each clause $C_j$.
All the numbers in set $S$ and $s$ are in base 10.
Besides, $10^{n+2m} < y_i < 2 \cdot 10^{n+2m}$, so that the bounded constraints are satisfied.
All the details regarding $S$ and $s$ are defined as follows.
\begin{itemize}
    \item
    In the set $S$, all integer numbers are with $n+2m+1$ digits, and the first digit are always 1.

    \item
    In the set $S$, we construct two integer numbers $t_i$ and $f_i$ for the variable $y_i$.
    For both of the values, the $n$ digits after the first `1' serve to indicate the corresponding variable in $S$.
    That is, the $i^{th}$ digit in these $n$ digits is set to 1 and all others are 0.
    For the next $m$ digits, the $j^{th}$ digit is set to 1 if the clause $C_j$ contains the respective literal.
    The last $m$ digits are always 0.

    \item
    In the set $S$, we also construct three integer numbers $c_{j1}, c_{j2}$ and $c_{j3}$ for each clause $C_j$.
    In $c_{jk}$ where $k = \{1, 2, 3\}$, the first $n$ digits after the first `1' are 0, and in the next $m$ digits all are 0 except the $j^{th}$ index setting to $k$.
    The last $m$ digits are all 0 except the $j^{th}$ index setting to 1.

    \item $T= (n+m)\cdot 10^{n+2m} + s_0$, where $s_0$ is an integer number with $n+2m$ digits.
    The first $n$ digits of $s_0$ are 1, in the next $m$ digits all are 4, and in the last $m$ digits all are 1.
\end{itemize}

Based on the above rules, given the $\mathsf{3SAT}$ instance in Eqn. (\ref{eq:ex_3sat}) the constructed set $S$ and target $s$ are shown in Fig. \ref{fig:ex_bss}.
Note that the highest digit achievable is 9, meaning that no digit will carry over and interfere with other digits.
\begin{figure}[htb!]
\centering
\begin{tabular}{ccccccccccc}
           &   &     &$y_1$ &$y_2$ &$y_3$ &$y_4$   &$C_1$  &$C_2$ &$A_1$  &$A_2$   \\
  \hline                                                                            
  $t_1$    &=  &1    &1     &0     &0      &0      &1      &0     &0      &0       \\
  $f_1$    &=  &1    &1     &0     &0      &0      &0      &1     &0      &0       \\
  $t_2$    &=  &1    &0     &1     &0      &0      &0      &1     &0      &0       \\
  $f_2$    &=  &1    &0     &1     &0      &0      &0      &0     &0      &0       \\
  $t_3$    &=  &1    &0     &0     &1      &0      &0      &0     &0      &0       \\
  $f_3$    &=  &1    &0     &0     &1      &0      &1      &0     &0      &0       \\
  $t_4$    &=  &1    &0     &0     &0      &1      &0      &0     &0      &0       \\
  $f_4$    &=  &1    &0     &0     &0      &1      &1      &1     &0      &0       \\
  $c_{11}$ &=  &1    &0     &0     &0      &0      &1      &0     &1      &0       \\
  $c_{12}$ &=  &1    &0     &0     &0      &0      &2      &0     &1      &0       \\
  $c_{13}$ &=  &1    &0     &0     &0      &0      &3      &0     &1      &0       \\
  $c_{21}$ &=  &1    &0     &0     &0      &0      &0      &1     &0      &1       \\
  $c_{22}$ &=  &1    &0     &0     &0      &0      &0      &2     &0      &1       \\
  $c_{23}$ &=  &1    &0     &0     &0      &0      &0      &3     &0      &1       \\
  \hline                                                                           
  $s$      &=  &6    &1     &1     &1      &1      &4      &4     &1      &1       \\
  $s_0$    &=  &     &1     &1     &1      &1      &4      &4     &1      &1       \\
\end{tabular}
\caption{The constructed $\mathsf{BSS}$ instance for the given $\mathsf{3SAT}$ instance in (\ref{eq:ex_3sat}).}
\label{fig:ex_bss}
\vspace{-.1in}
\end{figure}

\begin{myclaim}
The $\mathsf{3SAT}$ instance has a satisfying truth assignment iff the constructed $\mathsf{BSS}$ instance has a subset that adds up to $s$.
\end{myclaim}

{\bf Proof of $\Rightarrow$ part of Claim}:
If the $\mathsf{3SAT}$ instance has a satisfying assignment,
we can pick a subset containing all $t_i$ for which $y_i$ is set to true and $f_i$ for which $y_i$ is set to false.
We should then be able to achieve $s$ by picking the necessary $c_{jk}$ to get 4's in the $s$.
Due to the last $m$ `1' in $s$, for each $j \in [m]$ only one would be selected from $\{c_{j1}, c_{j2}, c_{j3}\}$.
Besides, we can see totally $n+m$ numbers would be selected from $S$.

{\bf Proof of $\Leftarrow$ part of Claim}:
If there is a subset $S' \in S$ that adds up to $s$, we will show that it corresponds to a satisfying assignment in the $\mathsf{3SAT}$ instance.
$S'$ must include exactly one of $t_i$ and $f_i$, otherwise the $i$th digit value of $s_0$ cannot be satisfied.
If $t_i \in S'$, in the $\mathsf{3SAT}$ we set $y_i$ to true; otherwise we set it to false.
Similarly, $S'$ must include exactly one of $c_{j1}, c_{j2}$ and $c_{j3}$, otherwise the last $m$ digits of $s$ cannot be satisfied.
Therefore, all clauses in the $\mathsf{3SAT}$ are satisfied and $\mathsf{3SAT}$ has a satisfying assignment.

\end{proof}

For instance, given a satisfying assignment of Eqn. (\ref{eq:ex_3sat}): $\langle y_1=0, y_2=1, y_3=0, y_4=0 \rangle$,
the corresponding subset $S'$ is $\{f_1=110000100, t_2=101000100, f_3=100101000, f_4=100011100, c_{12}=100002010, c_{21}=100000101\}$.
We set $s = (m+n) \cdot 10^{n+2m} + s_0$, where $s_0 = 11114411$, and then $s = 611114411$.
We can see that $f_1+t_2+f_3+f_4+s_{12}+s_{21} = s$.

Combining Lemma \ref{lem:np_bss} and Lemma \ref{lem:nph_bss}, we can achieve the following theorem.

\end{document}